%% file: Perspective_of_Visual_Sphere-JM_Fober.tex
\documentclass{jmfs}

\graphicspath{ {./fig/image/}{./fig/vector/} }
\newcommand{\maintitle}
{Perspective picture from Visual Sphere}
\newcommand{\subtitle}
{a new approach to image rasterization}
\newcommand{\license}
{https://creativecommons.org/licenses/by-nd/3.0/}

\hypersetup{
	pdftitle={\maintitle},
	pdfauthor={Jakub Maksymilian Fober},
	pdfsubject={\subtitle},
	pdfkeywords={perspective;non-linear projections;spherical perspective;graphics hardware;mathematics of art},
	hidelinks
}

\begin{document}

%%%%%%%%%%%%%%%%
%% Title Page %%
%%%%%%%%%%%%%%%%

\title{\maintitle: \\ \Large\textit{\subtitle}}
\author{Jakub Maximilian Fober
	\thanks{
		\href{mailto:talk@maxfober.space?subject=About\%20Perspective\%20Paper&cc=jakub.m.fober@pm.me}
		{talk@maxfober.space}
	\newline
		\href{https://maxfober.space}
		{\hspace{1.37em}$^*$https://maxfober.space}
	\newline
		\href{https://orcid.org/0000-0003-0414-4223}
		{\hspace{1.37em}$^*$https://orcid.org/0000-0003-0414-4223}
	}
}
% first revision publication date
% \date{\textit{January 7, 2020}}
% second revision publication date
% \date{\textit{March 12, 2020}}
% \date{\textit{April 10, 2020}}
\date{\textit{\today}}

% create title page
\label{pg:title}
\maketitle

% Perspective Map rasterization preview
\begin{figure}[h]
	\centerline{
		\includegraphics[width=5.8cm]{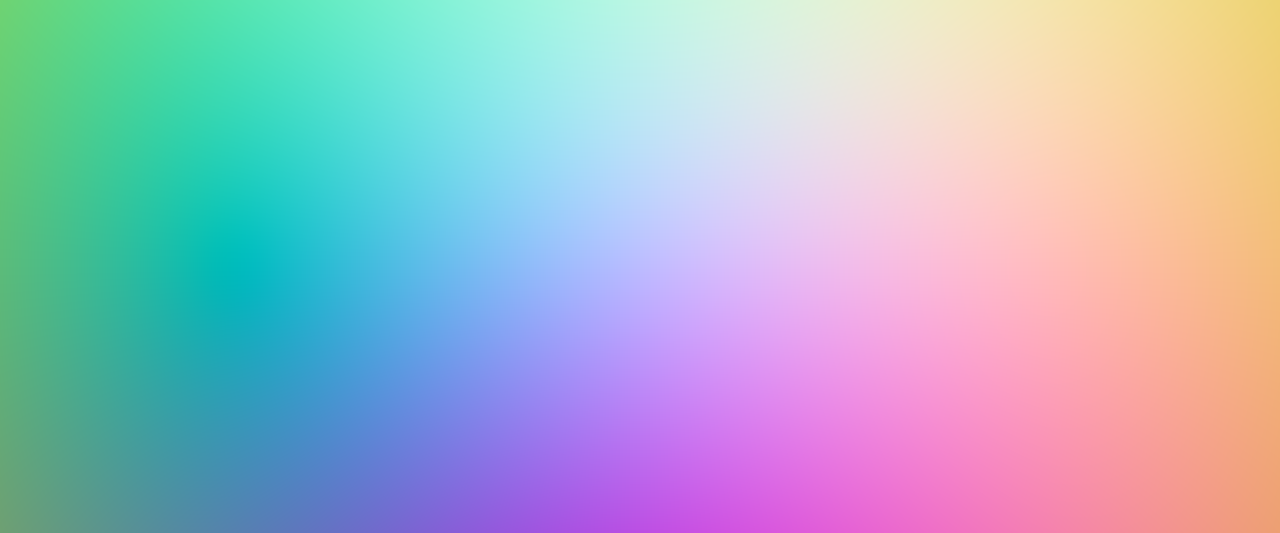}
		\hfill$\rightarrow$\hfill
		\includegraphics[width=5.8cm]{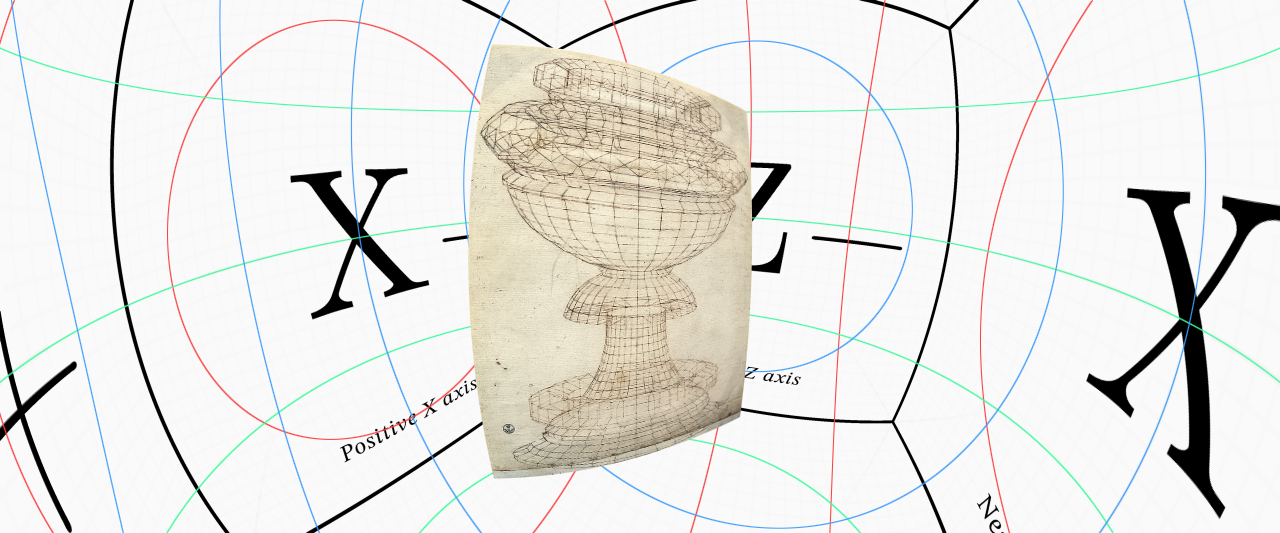}
	}
	\caption[Rasterized polygon example]
		{\small Example of aliasing-free, single-pass rasterization of polygon quad in real-time, from perspective vector map, where $\Omega^d=270\degree$, $k=0.32$, $l=62\%$, $s=86\%$.}
	\label{fig:title page}
\end{figure}

\begin{abstract}
	In this paper alternative method for real-time 3D model rasterization is given. Surfaces are drawn in perspective-map space which acts as a virtual camera lens. It can render single-pass 360\textdegree\ angle of view (AOV) image of unlimited shape, view-directions count and unrestrained projection geometry (e.g. direct lens distortion, projection mapping, curvilinear perspective), natively aliasing-free. In conjunction to perspective vector map, visual-sphere perspective model is proposed. A model capable of combining pictures from sources previously incompatible, like fish-eye camera and wide-angle lens picture. More so, method is proposed for measurement and simulation of a real optical system variable no-parallax point (NPP). This study also explores philosophical and historical aspects of picture perception and presents a guide for perspective design.
\end{abstract}

\vfill
%%%%%%%%%%%%%%%%%%%
%% Small license %%
%%%%%%%%%%%%%%%%%%%
% Creative Commons BY-ND 3.0 logo
\hyperlink{license}{\includegraphics{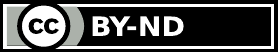}}
\pagebreak

%%%%%%%%%%%%%%%%%%
%% Introduction %%
%%%%%%%%%%%%%%%%%%

\section*{Introduction} % unnumbered unlisted section

There is a great demand for perspective projection model able to produce computer-generated (CG) image up to 360\textdegree\ of view, with lens-distortions, directly from three-dimensional space, to pixel data. Currently there is no practical direct method for rasterization of real-time graphics in curvilinear perspective.\supercite{Muszynski2019Shaders} Every real-time perspective imagery incorporates \emph{Pinhole Camera} model as a base, some with additional layers of distortion on top. Also knowledge about relationship between motion and perspective has not been properly formulated, leaving void in that field of image science.
This paper aims at solving those issues. Study presents new, universal model for perspective projection and rasterization in CG graphics. It also explores history of perspective picture, redefines abstract theorem of image (as recorded in common-knowledge) and establishes rules for picture's perspective design.

This paper begins with a brief introduction to the visual-sphere perspective and jumps right into technical details of rasterization. If one is not familiar with the topic or terms, \emph{Appendix} \vpageref{sec:appendix} gives great introduction and presents some basic knowledge on picture, perspective and perception, along with brief history of the topic.
If one is familiar with this work and looks for some reference, Section \vref{sec:transformations} presents algorithms of various perspective maps, used in visual-sphere rasterization.
Within that, subsection \vref{sub:universal} presents \emph{Universal Perspective} model. It describes variable image geometry in common projections.
Variable (or floating) no-parallax point of real optical systems is described in section \vref{sec:npp}.
Additional figures are presented at the end of the document.
At the very end of this document are code listings in GLSL for most of presented equations.

If one is looking for some artistic view on the topic, I recommend \emph{Conclusion} \vpageref{sec:conclusion} and some content from \emph{Appendix} with Chapter's 1 introduction.

I hope that this work will bring a new breath into image production, open minds and possibilities for creative action and thinking.
\pagebreak

%%%%%%%%%%%%%%%
%% Bookmarks %%
%%%%%%%%%%%%%%%

% Contents list
\phantomsection
\pdfbookmark[1]{Table of contents}{sec:contents} % add to bookmarks
\tableofcontents

% \pagebreak

% Figures list
\phantomsection
\pdfbookmark[1]{Figures list}{sec:figures} % add to bookmarks
\listoffigures

% \pagebreak

% Code listings list
\phantomsection
\pdfbookmark[1]{Code list}{sec:listings} % add to bookmarks
\lstlistoflistings

\pagebreak

%%%%%%%%%%%
%% Paper %%
%%%%%%%%%%%

\section[on Visual Sphere perspective]{From Visual Pyramid to Visual Sphere}
\label{sec:visual sphere}

% Visual Sphere model
\begin{figure}[H]
	\centering\includegraphics{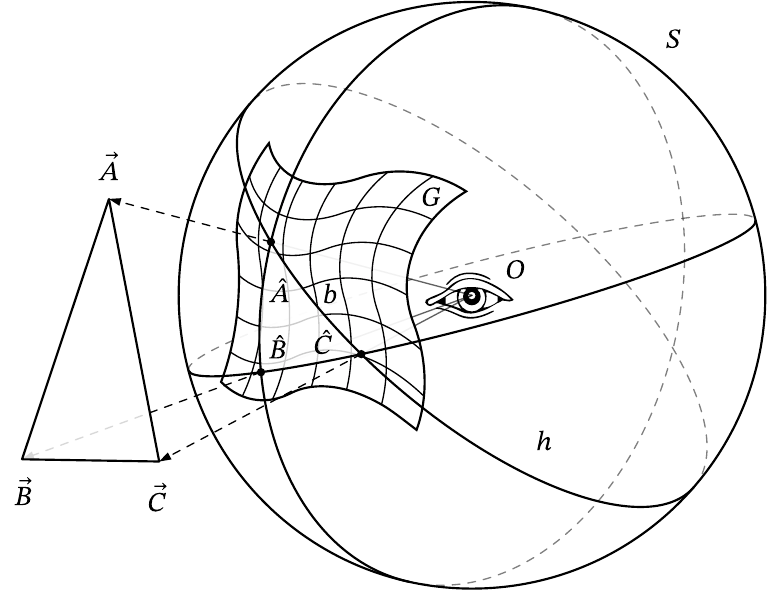}
	\caption[Visual sphere model]
		{Projection of triangle $\overline{ABC}$ onto visual unit sphere $S$, where projection origin $O$ is at the sphere center. Edge of the projected triangle is always produced by an arc of a great circle, here $\overline{AC}\mapsto b\in h$. Grid $G$ represents visual sphere vector map, where each pixel color is a spherical unit vector.}
	\label{fig:Visual-sphere-graph}
\end{figure}

Visual pyramid of \noun{Alberti} theorem\supercite{Alberti1972teatrise} is by definition restricted to acute angles, which is limiting in terms of describable projections. This property makes stitching or layering multiple pictures defined in such space a problematic task.
In the standard perspective model, introduced by \noun{L.B. Alberti}, 3D point position is transformed into 2D plane/screen coordinates. Noticeably in some curvilinear projections, points are stretched into lines (see example \vref{ex:nonlinear mapping}), making single visible point appear in multiple spots on picture plane/screen. This makes linear 3D$\mapsto$2D point position mapping insufficient.

In the proposed visual-sphere model, every visible point in perspective picture's space has its own spherical position coordinate. Thus single point of a visual sphere can occupy multiple places in the picture, which conforms to this principle of curvilinear perspective. Such perspective space (a sphere) allows for stitching, layering and mapping of images in any single-point projection and perspective geometry.

Points in spherical projection model are no longer transformed into screen space. Rather lines are calculated and combined forming a polygon image (see figure \vref{fig:Visual-sphere-graph}). Straight line projected through the center point $O$ will always form an arc of a great circle $h$. Such great circle can be mapped onto visual sphere described by incident vector map (later referred as \emph{perspective map}). This process involves rotating the perspective map vector data. Goal is to align one of the axis of the map $G$ with a great circle forming polygon edge (see figure \ref{fig:Visual-sphere-graph}). One way to rotate axis component is to calculate dot product between the perspective vector map $G$ and a unit vector perpendicular to points forming polygon edge, like $\Vert\vec C\times\vec A\Vert$ and the great circle $h$.

This process is essentially a half-space rasterization technique,\supercite{Pineda1988EdgeFunctions} extended to spherical space.

\subsection[Rasterization with perspective map]{Rasterization of the $\overline{ABC}$ polygon triangle using $\hat G$ vector from perspective map}
\label{sub:rasterization}

Projected polygon geometry is always part of a great circle. The goal of the algorithm is to rasterize polygon shape formed by those spherical lines. Rasterization process involves determining orientation of the great-circle. Then rotating perspective map, so that the vertical axis of the map aligns with a great-circle. Next the step-function is performed on a $\vec G'_i$ component of the rotated vector map $G'$.
The step function algorithm is essential for aliasing-free edge rasterization.
Full $n\text{-sided}$ convex polygon picture is defined by intersection of $n\text{-number}$ of such operations.
\begin{subequations}
\begin{gather}
	\begin{bmatrix}
		\vec G'_1 \\
		\vec G'_2 \\
		\vec G'_3
	\end{bmatrix}
	=
	\begin{bmatrix}
		\begin{aligned}
			\hat G &\cdot \Vert\vec A\times\vec B\Vert \\
			\hat G &\cdot \Vert\vec B\times\vec C\Vert \\
			\hat G &\cdot \Vert\vec C\times\vec A\Vert
		\end{aligned}
	\end{bmatrix} \\
	\overline{ABC}\mapsto G = \min\left\{ \text{step}(\vec G'_1),\; \text{step}(\vec G'_2),\; \text{step}(\vec G'_3) \right\} \qed
\end{gather}
\end{subequations}
Rotation of perspective map vector $\hat G$ is performed with a dot product between each $\hat G$ and rotation-direction vector. Which is a dot product between perspective map $G$ and normalized cross product of two edge points.
Rotation vector doesn't have to be normalized, but that could cause some precision errors, especially when using aliasing-free step function.

% visual sphere coordinates
\begin{figure}[H]
	\centerline{
		\begin{subfigure}[t]{90pt}
			\includegraphics{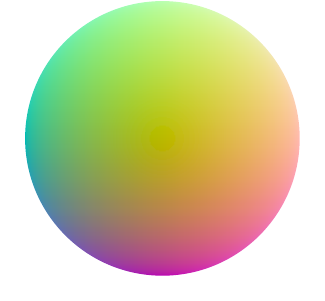}
			\caption{$RGB$ as $xyz$ coordinates}
		\end{subfigure}
		$\leftarrow$
		\begin{subfigure}[t]{90pt}
			\includegraphics{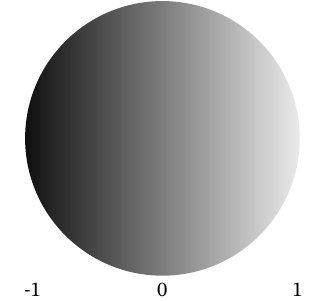}
			\caption{Unit $\text{vector-}x$}
		\end{subfigure}
		\phantom{$\leftarrow$}
		\begin{subfigure}[t]{90pt}
			\includegraphics{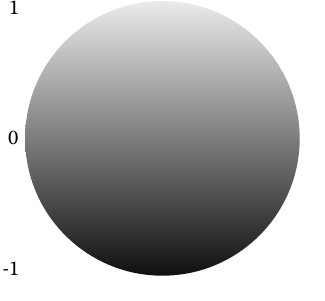}
			\caption{Unit $\text{vector-}y$}
		\end{subfigure}
		\phantom{$\leftarrow$}
		\begin{subfigure}[t]{90pt}
			\includegraphics{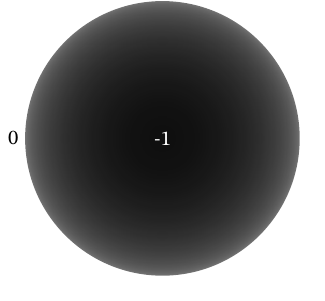}
			\caption{Unit $\text{vector-}z$}
		\end{subfigure}
	}
	\caption[Visual sphere coordinates]
		{Visual sphere in vector coordinates as seen from outside, facing $Z$ direction. Each vector component $x$, $y$ and $z$ can be rotated independently.}
	\label{fig:visual sphere coordinates}
\end{figure}

\noindent Visual sphere perspective model is ultimately a \emph{cosine}-based perspective model and can produce pictures in both \emph{tangent perspective} and \emph{sine perspective}.
As seen in figure \vref{fig:visual sphere coordinates}, each sphere surface point has a value associated to its position. Those three values also represent cosine of an angle between axis vectors $[1\quad0\quad0]$, $[0\quad1\quad0]$, $[0\quad0\quad1]$ and the surface position vector. After rotation of $\hat G$, those three axes are represented by rotation matrix vectors.

\subsubsection{Aliasing-free step function}
\label{sub:step function}

Step function is performed on rotated $\vec G'$ coordinate and determines inside and outside of a polygon for each edge. Aliasing-free result can be achieved by changing width of the step-function slope from binary, to one-pixel wide (see figure \vref{fig:step function}).
Below presented are two algorithms for a step function.
% step function
\begin{subequations}
\begin{align}
	\text{pstep}\big(g\big) &= \left\{\frac{g}{\partial\big(g\big)}+\frac{1}{2}\right\}\cap [0,1] \\
	\text{bstep}\big(g\big) &=
		\begin{cases}
			1, & \text{if}\ g>0 \\
			0, & \text{otherwise}
		\end{cases}
\end{align}
\end{subequations}
$\text{pstep}(g)$ function (\emph{p} standing for \emph{pixel}) gives smooth aliasing-free result. The $\text{bstep}(g)$ function is a binary operation and gives aliased output. $\partial(x)$ function is an equivalent to $\textbit{fwidth}(x)$ GLSL function. $\nicefrac{1}{2}$ offset centers anti-aliasing gradient at the polygon edge (as seen in subfigure \ref{fig:step clamp} compared to subfigure \vref{fig:step gradient}).

\paragraph{Global delta $\delta$} can replace denominator $\partial(g)$ in $\text{pstep}()$ function. It gives more natural look and has smaller computational footprint.
Global delta $\delta$ is derived from perspective vector map $\hat G$.
Its value can be preprocessed and saved in a texture.
The difference between preprocessed delta $\delta$ and per-fragment $\partial$ is in polygon-edge slope width.

\begin{subequations}
\begin{align}
	\delta &=
		\frac{1}
		{\max\big\{
			\partial\big(\hat G_x\big),
			\partial\big(\hat G_y\big),
			\partial\big(\hat G_z\big)
		\big\}} \\
	\text{gpstep}\big(g\big) &= \left\{\delta g + \frac{1}{2} \right\} \cap [0,1]
\end{align}
\end{subequations}
The result of $\delta$ calculation is a single-channel texture, which represents reciprocal of the component-wise maximum $\textbit{fwidth}(x)$ of perspective map vector $\hat G$.

% step function
\begin{figure}[H]
	\centerline{
		\begin{subfigure}[t]{120pt}
			\includegraphics{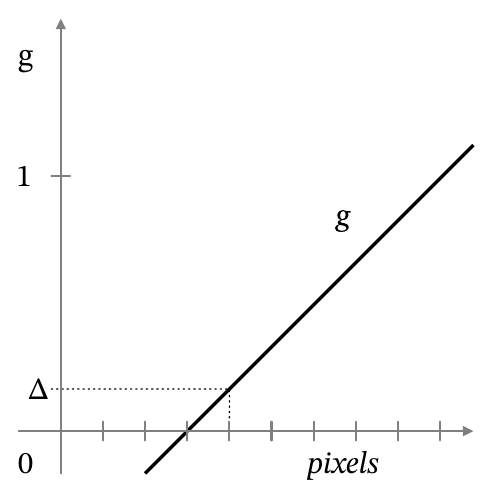}
			\caption{Gradient before the step function.}
			\label{fig:step gradient}
		\end{subfigure}\hspace{2em}
		\begin{subfigure}[t]{120pt}
			\includegraphics{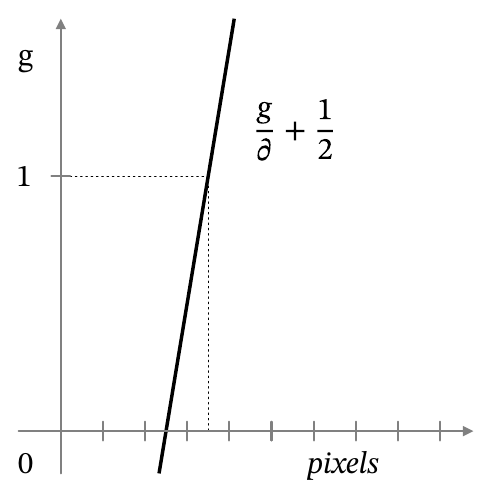}
			\caption{Gradient divided by its pixel $\partial$.}
			\label{fig:step delta}
		\end{subfigure}\hspace{2em}
		\begin{subfigure}[t]{120pt}
			\includegraphics{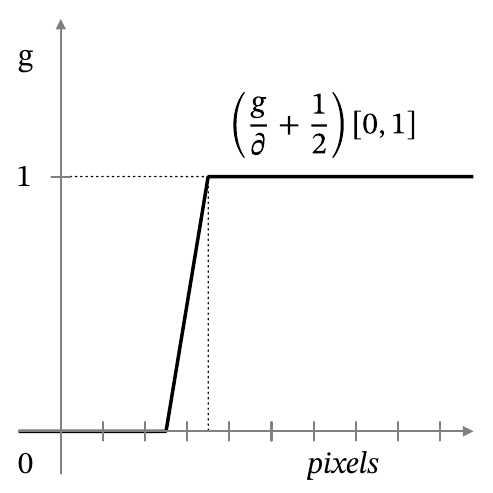}
			\caption{Pixel gradient clamped to $[0,1]$ range.}
			\label{fig:step clamp}
		\end{subfigure}
	}
	\caption[Aliasing-free step function]
		{Aliasing-free step function process, where $g$ is the gradient function. $\partial$ is equivalent of $\text{fwidth}(g)$ GLSL function. Horizontal axis represents pixel position and vertical, the value of $g$.}
	\label{fig:step function}
\end{figure}

% pixel step function
\begin{lstlisting}[float=p,
	label={lst:pixel step},
	caption={
		[Pixel step function]
		Aliasing-free step function in GLSL.
	}
]
// Pixel step
float pxstep(float gradient)
{ return clamp(gradient/fwidth(gradient)+0.5, 0.0, 1.0); }

// Binary pixel step
bool bpxstep(float gradient)
{ return gradient > -fwidth(gradient)*0.5; }

// Partial derivative calculation, can be saved as a texture
float get_map_delta(vec3 visual_sphere)
{
	visual_sphere = fwidth(visual_sphere);
	// Reciprocal of perspective map maximum partial derivative
	return 1.0/max(max(visual_sphere.x, visual_sphere.y), visual_sphere.z);
}

// Pixel step with derivative in texture
float gpxstep(float gradient, float map_del)
{ return clamp(gradient*map_del+0.5, 0.0, 1.0); }

// Binary pixel step with derivative in texture
bool gbpxstep(float gradient, float map_del)
{ return gradient > -0.5/map_del; }
\end{lstlisting}

\begin{rem}
	When combining polygons into polygon-strip using pixel-step function, it is important to sum each mask, otherwise visible seam may occur.
	See more on combining fragments into buffer in subsection \vref{sub:occlusion}.
\end{rem}

\subsubsection[Miter mask for AA]{Miter mask for pixel-step rasterization}
\label{sub:miter mask}

In special case, when projected polygon's edges meet at very shallow angle, corners will extend beyond polygon outline (due to half-pixel offset in the $\text{pstep}(g)$ function). This visual artifact can be corrected by a miter mask. There are many ways to calculate such mask, one is to define the smallest circle over projected triangle. Following algorithm uses barycentric coordinates to determine the smallest circle's center and size.\supercite{Ericson2007SmallestCircle}

\begin{subequations}
\begin{align}
	&\phantom= \begin{cases}
		\begin{aligned}
			a^2	&= |\hat B-\hat C|^2 = (\hat B-\hat C)\cdot(\hat B-\hat C) \\
			b^2	&= |\hat C-\hat A|^2 = (\hat C-\hat A)\cdot(\hat C-\hat A) \\
			c^2	&= |\hat A-\hat B|^2 = (\hat A-\hat B)\cdot(\hat A-\hat B)
		\end{aligned}
	\end{cases} \\
	\begin{bmatrix}
	 \vec O_s \\
	 \vec O_t \\
	 \vec O_p
	\end{bmatrix} &= \begin{bmatrix}
	 a^2(b^2+c^2-a^2) \\
	 b^2(c^2+a^2-b^2) \\
	 c^2(a^2+b^2-c^2)
	\end{bmatrix} \\
	\vec S	&=
	\begin{cases}
		0.5(\hat B+\hat C),	& \text{if}\quad \vec O_s \leq 0 \\
		0.5(\hat C+\hat A),	& \text{if}\quad \vec O_t \leq 0 \\
		0.5(\hat A+\hat B),	& \text{if}\quad \vec O_p \leq 0 \\
		\frac{\vec O_s\hat A+\vec O_t\hat B+\vec O_p\hat C}{\vec O_s+\vec O_t+\vec O_p},	& \text{otherwise} \\
	\end{cases} \qed
\end{align}
\end{subequations}
Where $\vec O$ is the barycentric coordinate of circumcenter. $a^2,b^2,c^2$ are squared lengths of projected triangle's edges. $\vec S$ is the smallest-circle center vector which length is equal to the cosine of an angle between $\hat S$ and the smallest-circle rim. Polygon triangle is degenerate if $0=\vec O_s+\vec O_t+\vec O_p$, meaning all projected $\hat A,\hat B,\hat C$ points lay in one line. In such case miter mask can be omitted.
% smallest circle function
\begin{lstlisting}[float=p,
	label={lst:miter mask},
	caption={
		[Miter mask for aliasing-free rasterization]
		Polygon miter mask function in GLSL, where matrix \textbit{triangle} represents vertices in camera-space.
	}
]
float dot(vec3 vector) { return dot(vector, vector); }

// Once per polygon
vec4 getMiterVector(mat3 triangle)
{
	// Project polygon onto sphere
	triangle[0] = normalize(triangle[0]); // Vertex A
	triangle[1] = normalize(triangle[1]); // Vertex B
	triangle[2] = normalize(triangle[2]); // Vertex C

	// Get polygon sides squared
	vec3 barycenter = vec3(
		dot(triangle[1]-triangle[2]),
		dot(triangle[2]-triangle[0]),
		dot(triangle[0]-triangle[1])
	);
	// Get circumcenter
	barycenter *= barycenter.gbr+barycenter.brg-barycenter;
	// Get smallest circle vector
	if (barycenter.s<=0.0)
		barycenter = 0.5*(triangle[1]+triangle[2]);
	else if (barycenter.t<=0.0)
		barycenter = 0.5*(triangle[2]+triangle[0]);
	else if (barycenter.p<=0.0)
		barycenter = 0.5*(triangle[0]+triangle[1]);
	else
		barycenter = (barycenter.s*triangle[0]+barycenter.t*triangle[1]+barycenter.p*triangle[2])
			/(barycenter.s+barycenter.t+barycenter.p);

	// Smallest circle center vector and radius
	return vec4(normalize(barycenter), length(barycenter) );
}
\end{lstlisting}

\noindent Having the smallest circle center vector $\vec S$, rasterization algorithm can be updated as follows:
\begin{subequations}
\begin{gather}
	\begin{bmatrix}
		\vec G'_1 \\
		\vec G'_2 \\
		\vec G'_3 \\
		\vec G'_4
	\end{bmatrix}
	=
	\begin{bmatrix}
		\begin{aligned}
			\hat G &\cdot \Vert\vec A\times\vec B\Vert \\
			\hat G &\cdot \Vert\vec B\times\vec C\Vert \\
			\hat G &\cdot \Vert\vec C\times\vec A\Vert \\
			\hat G &\cdot \hat S - |\vec S|
		\end{aligned}
	\end{bmatrix} \label{eq:rasterization}\\
	\overline{ABC}\mapsto G = \min\left\{\text{pstep}(\vec G'_1),\; \text{pstep}(\vec G'_2),\; \text{pstep}(\vec G'_3),\; \text{pstep}(\vec G'_4) \right\} \qed
\end{gather}
\end{subequations}

\begin{figure}[H]
	\centering\includegraphics{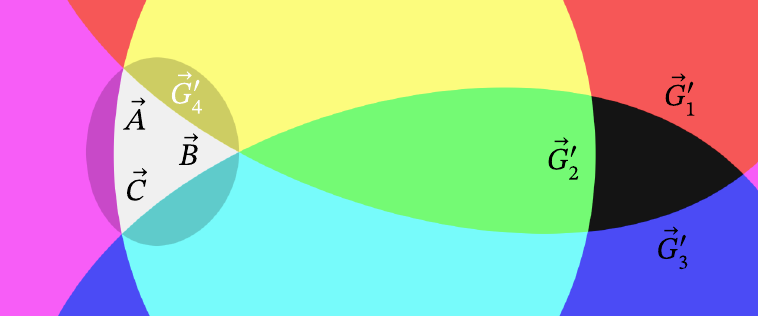}
	\caption[Rasterization masks in RGB]
		{Rasterization masks as RGB values, where perspective map represents $\Omega^d=270\degree$, $k=0.32$, $l=62\%$, $s=86\%$.}
	\label{fig:raster masks}
\end{figure}

\subsubsection[Matrix rasterization]{Rasterization using matrix multiplication}
\label{sub:n-sided polygon}

Since rasterization involves multiple dot products (equation \ref{eq:rasterization}), part of the process can be mitigated to matrix multiplication, where each matrix row represents rotation vector. Such matrix is calculated once per polygon--per frame and is executed per fragment pixel. Full vertex rasterization process is presented on figure \vref{fig:vertex flowchart}.

\begin{subequations}
\begin{gather}
	\begin{bmatrix}
		\vec G'_1 \\
		\vec G'_2 \\
		\vec G'_3 \\
		\vec G'_4
	\end{bmatrix}
	=
	\begin{bmatrix}
		\hat G_x \\
		\hat G_y \\
		\hat G_z
	\end{bmatrix}
	\begin{bmatrix}
		\Vert\vec A\times\vec B\Vert \\
		\Vert\vec B\times\vec C\Vert \\
		\Vert\vec C\times\vec A\Vert \\
		\hat S_x\ \hat S_y\ \hat S_z
	\end{bmatrix}
	-
	\begin{bmatrix}
		0 \\
		0 \\
		0 \\
		|\vec S|
	\end{bmatrix}
\\
	\overline{ABC}\mapsto G = \min\left\{\text{pstep}(\vec G'_1),\; \text{pstep}(\vec G'_2),\; \text{pstep}(\vec G'_3),\; \text{pstep}(\vec G'_4) \right\} \qed
\end{gather}
\end{subequations}
\begin{lstlisting}[float=p,
	lastline=27,
	label={lst:rasterization},
	caption={
		[Aliasing-free rasterization]
		Aliasing-free rasterization function in GLSL. Function $\textbit{pxstep}()$ is described in listing \vref{lst:pixel step}.
	}
]
vec3 normal(vec3 A, vec3 B)
{ return normalize(cross(A, B) ); }

// Once per polygon
mat3 getRasterMatrix(mat3 triangle)
{
	return mat3(
		normal(triangle[0], triangle[1]),
		normal(triangle[1], triangle[2]),
		normal(triangle[2], triangle[0])
	);
}

// Once per fragment pixel
float getPolygonOutline(vec3 visual_sphere, mat3 raster, vec4 miter)
{
	vec4 masks;
	masks.rgb = visual_sphere*raster;
	masks.a = dot(visual_sphere, miter.xyz)-miter.w;

	float outline = pxstep(masks[0]);
	for (int i=1; i<4; i++)
		outline = min(pxstep(masks[i]), outline);

	return outline;
}
\end{lstlisting}
\begin{lstlisting}[float=p,
	firstline=28,
	label={lst:aliased rasterization},
	caption={
		[Aliased (jagged) rasterization]
		Binary (jagged) step rasterization function in GLSL. This version produces aliased result, suitable for low-resolution fragment region evaluation.
	}
]
// Once per region
bool getPolygonRegion(vec3 visual_sphere, mat3 raster)
{
	vec3 masks = visual_sphere*raster;
	// extend outline
	masks += fwidth(masks)*0.5;

	bool outline = masks[0]>0.0;
	for (int i=1; i<4; i++)
		outline &= masks[i]>0.0;

	return outline;
}
\end{lstlisting}

\noindent This procedure can be expanded to any convex $n\text{-sided}$ planar polygon using procedural equation as follows:

\begin{subequations}
\begin{align}
	\begin{bmatrix}
		\vec G'_1 \\
		\vec G'_2 \\
		\vdots \\
		\vec G'_n \\
	\end{bmatrix}
	&=
	\begin{bmatrix}
		\hat G_x \\
		\hat G_y \\
		\hat G_z
	\end{bmatrix}
	\begin{bmatrix}
		\Vert T_{\vec 1} \times T_{\vec 2} \Vert \\
		\Vert T_{\vec 2} \times T_{\vec 3} \Vert \\
		\vdots \\
		\Vert T_{\vec n} \times T_{\vec{(i}+1)\mod n} \Vert \\
	\end{bmatrix}
	\\
	\text{mask}(\vec G') &= \min\left\{ \text{step}(\vec G'_1),\; \text{step}(\vec G'_2), \hdots, \text{step}(\vec G'_n) \right\} \qed
\end{align}
\end{subequations}
Here $T$ describes $n\text{-sided}$ polygon points $n\times3$ matrix, where each row is a vertex position (counting clock-wise). Each rotation vector is derived from cross product between $i\text{-vertex}$ and $i+1$. For the the last $n\textsuperscript{th}$ vertex, next one is \textnumero 1.

\subsubsection[Interpolation in Fragment Shader]{Fragment data interpolation from barycentric coordinates}
\label{sub:barycentric}

Rendering realistic polygon graphics involves shading and texture mapping. Values of normal, depth and UV coordinates associated to each vertex are interpolated across polygon surface using barycentric coordinate of the fragment point.

\begin{figure}[H]
	\centering\includegraphics{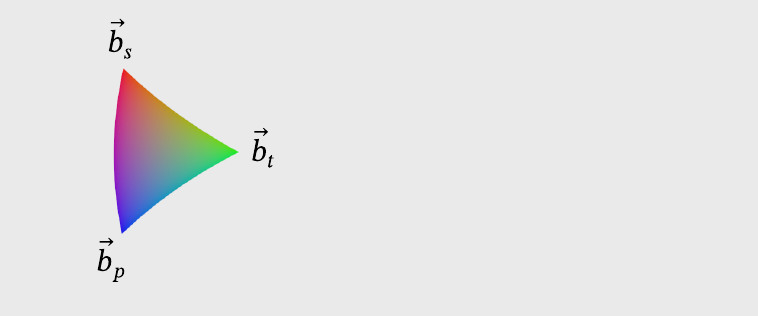}
	\caption[Barycentric coordinates in RGB]
		{Barycentric coordinates as RGB values. Here perspective map represents $\Omega^d=270\degree$, $k=0.32$, $l=62\%$, $s=86\%$.}
	\label{fig:barycentric rgb}
\end{figure}

\begin{equation}
	\vec N = \begin{cases}
		(\vec A-\vec B)\times(\vec C-\vec B) \\
		(\vec B-\vec C)\times(\vec A-\vec C) \\
		(\vec C-\vec A)\times(\vec B-\vec A)
	\end{cases}
\end{equation}
Normal vector $\vec N\in\mathbb{R}^3$ of the triangle plane $\overline{ABC}$ is derived from cross product between two triangle edges. Length of this vector is equal to the area of a parallelogram formed by those two edges, which is equal to double the area of triangle $\overline{ABC}$.

\begin{subequations}
\begin{align}
	r 	&= \frac{\vec A \cdot \vec N}{\hat G \cdot \vec N}
		= \frac{\vec B \cdot \vec N}{\hat G \cdot \vec N}
		= \frac{\vec C \cdot \vec N}{\hat G \cdot \vec N} \\
		&\equiv |\hat G\rightarrow\overline{ABC}|
\end{align}
\end{subequations}
Distance $r$ represents multiplier of the visual sphere vector $\hat G$, to intersection point on the $\overline{ABC}$ triangle plane. Since $\hat G$ is a unit vector, value $r$ can be exported as depth (representing distance, not $z$--position). Here vector $\vec A$, $\vec B$, $\vec C$ in the numerator can be replaced by any point on the triangle plane.

\begin{subequations}
\begin{align}
	\begin{bmatrix}
		\vec b_s \\
		\vec b_t \\
		\vec b_p
	\end{bmatrix}
	&=
	\begin{bmatrix}
		|\big(\vec B-r\hat G\big) \times \big(\vec C-r\hat G\big)|\\
		|\big(\vec C-r\hat G\big) \times \big(\vec A-r\hat G\big)|\\
		|\big(\vec A-r\hat G\big) \times \big(\vec B-r\hat G\big)|
	\end{bmatrix} \div |\vec N|
\end{align}
\end{subequations}
Since $\vec A\cdot\vec B=|\vec A||\vec B|\cos\alpha$ and $\overline{AB}\ \bot\ (\vec A\times\vec B)$ this equation can be rewritten using triple product, as follows.
\begin{subequations}
\begin{align}
	\begin{bmatrix}
		\vec b_s \\
		\vec b_t \\
		\vec b_p
	\end{bmatrix}
	&=
	\begin{bmatrix}
		\big(\vec C-r\hat G\big) \times \big(\vec B-r\hat G\big) \cdot \vec N \\
		\big(\vec A-r\hat G\big) \times \big(\vec C-r\hat G\big) \cdot \vec N \\
		\big(\vec B-r\hat G\big) \times \big(\vec A-r\hat G\big) \cdot \vec N
	\end{bmatrix} \div \big(\vec N \cdot \vec N\big)
\end{align}
\end{subequations}

% barycentric coordinates function
\begin{lstlisting}[float=p,
	label={lst:barycentric},
	caption={
		[Barycentric coordinate interpolation vector]
		Barycentric vector function with hard normal and depth function for fragment-data interpolation in GLSL.
	}
]
float dot(vec3 vector) { return dot(vector, vector); }

// Once per polygon
vec3 getHardNormal(mat3 triangle)
{ return cross(triangle[0]-triangle[1], triangle[2]-triangle[1]); }

// Once per fragment pixel
vec4 getBarCoord(vec3 visual_sphere, mat3 triangle, vec3 hard_normal)
{
	// Extend onto polygon plane intersection
	float depth = dot(triangle[0], hard_normal)/dot(visual_sphere, hard_normal);
	visual_sphere *= depth;

	for (int i=0; i<3; i++)
		triangle[i] -= visual_sphere;

	return vec4(vec3(
		dot(cross(triangle[2], triangle[1]), hard_normal),
		dot(cross(triangle[0], triangle[2]), hard_normal),
		dot(cross(triangle[1], triangle[0]), hard_normal)
	)/dot(hard_normal), depth);
}
\end{lstlisting}
Barycentric vector $\vec b$ is a proportion of surface area. From vector $\vec b$, various vertex properties are interpolated (e.g. depth, normal vector and texture coordinates), given each vertex $A$, $B$ and $C$ has an associated value.

\begin{subequations}
\begin{align}
	f_r &= r \\
	f_{\hat N} &= \Vert \vec b_sA_{\hat N} + \vec b_tB_{\hat N} + \vec b_pC_{\hat N} \Vert \\
	f_{\vec f} &= \vec b_sA_{\vec f} + \vec b_tB_{\vec f} + \vec b_pC_{\vec f}
\end{align}
\end{subequations}
Interpolation of fragment data $f$ is done through a dot product between barycentric coordinate vector $\vec b$ and values associated to each vertex. Here $f_r$ is the depth pass (representing distance, not $z$--position), $f_{\hat N}$ is the interpolated normal vector and $f_{\vec f}$ are the texture coordinates. All interpolations are perspective-correct. Figure \vref{fig:fragment flowchart} shows this process in a step-by-step flowchart.

% fragment interpolation function
% \lstinputlisting[float=p,
% 	label={lst:fragment interpolation},
% 	caption={
% 		[Fragment interpolation]
% 		Polygon triangle data structure and fragment data interpolation function in GLSL.
% 	}
% ]{./glsl/fragment_interp.glsl}

\subsubsection[Perspective shader pass]{Perspective pixel-shader pass}
\label{sub:perspective shader}

Prior to geometry rasterization, pixel perspective shader pass may be performed. Its output works as a perspective vector map. This pipeline addition enables special effects like dynamic perspective and projection mapping, flat mirror reflection, screen-space refraction (e.g. concave refractive surfaces which expand AOV), etc.
\begin{example}\label{ex:projection mapping}
	Projection mapping with dynamic view-position can be achieved by transforming vector data of world-position-pass texture $S$. Knowing viewer position $\vec O$, offset can be applied to $\vec S$. When normalized, $\hat S'$ produces perspective map vector $\hat G = \Vert \vec S-\vec O \Vert$. Complexity of projection surface and number of views (used projectors) is outside of concern, as view-position transformation is performed on a baked texture.
	Exact formula for projection mapping is available in sub-subsection \vref{sub:various projections}.
\end{example}

\subsection[Wire-frame rasterization]{Wire-frame line segment rasterization}
\label{sub:wireframe}

It is possible to produce screen-relative line drawing based on a perspective map.
Following algorithm will produce wire-frame image of projected $\overline{AB}$ line segment.
\begin{equation}
	\vec G'_z = \hat G\cdot\Vert\vec A\times\vec B\Vert
\end{equation}
Perspective vector map component $\hat G_z$ is rotated by $\overline{ABO}$ tangent vector $\Vert\vec A\times\vec B\Vert$, where $\vec O$ is the observation point.

\begin{subequations}
\begin{align}
	\text{blstep}(\vec G'_z) &=
		\begin{cases}
			1, & \text{if}\quad
					\partial\big(\vec G'_z\big) - |2\vec G'_z| > 0 \\
			0, & \text{otherwise}
		\end{cases}
	\label{eq:blstep}
	\\
	\text{plstep}(\vec G'_z) &= 1-\min\left\{\frac{|\vec G'_z|}{\partial\big(\vec G'_z\big)},\ 1\right\}
	\label{eq:plstep}
\end{align}
\end{subequations}
Above is a line-step function in two variants, binary-aliased (\ref{eq:blstep}) and pixel-smooth (\ref{eq:plstep}). $\partial(x)$ is equivalent to $\textbit{fwidth}(x)$ function of GLSL.

\begin{subequations}
\begin{align}
	\vec L &= \frac{\hat A+\hat B}{2} \\
	h &= \text{lstep}(\vec G'_z) \\
	l &= \min\left\{ \text{step}\big(\hat G\cdot\hat L-|\vec L|\big),\; h \right\} \qed
\end{align}
\end{subequations}
Radial mask combined with great circle $h$ forms $\overline{AB}$ line segment image $l$. \\
$\hat L$ is the line-middle vector, $\text{lstep}(\vec G'_z)$ function rasterizes great-circle $h$. Radial mask is formed by the $\text{step}(x)$ function of dot product between perspective map vector $\hat G$ and line-middle vector $\hat L$, minus $|\vec L|=\cos\nicefrac{\theta}{2}$.
% line segment function
\begin{lstlisting}[float=p,
	lastline=25,
	label={lst:line segment},
	caption={
		[Aliasing-free wireframe]
		Aliasing-free line segment rasterization function in GLSL. Function $\textbit{pxstep}()$ can be found in listing \vref{lst:pixel step}.
	}
]
float lpxstep(float scalar)
{ return 1.0-min(abs(scalar)/fwidth(scalar), 1.0); }

// Once per polygon side
vec4 getLineVector(vec3 A, vec3 B)
{
	// Project vertices onto sphere
	A = normalize(A); B = normalize(B);
	vec3 line_center = (A+B)*0.5;
	return vec4(normalize(line_center), length(line_center) );
}

// Once per polygon side
vec3 normal(vec3 A, vec3 B)
{ return normalize(cross(A, B) ); }

// Once per line's pixel
float getLineSegment(vec3 vis_sphere, vec3 normal, vec4 line_vec)
{
	// Great circle
	float line = lpxstep(dot(vis_sphere, normal) );
	// Apply segment mask
	return line*pxstep(dot(vis_sphere, line_vec.xyz)-line_vec.w);
}
\end{lstlisting}
\begin{lstlisting}[float=p,
	firstline=27,
	label={lst:aliased line segment},
	caption={
		[Aliased (jagged) wireframe]
		Binary (jagged) line segment rasterization function in GLSL.
	}
]
bool getAliasedLineSegment(vec3 vis_sphere, vec3 normal, vec4 line_vec)
{
	float great_circle = dot(vis_sphere, normal);
	// Binary line-step operation
	bool line = fwidth(great_circle)-abs(2.0*great_circle) > 0.0;
	// Apply segment mask
	return line && dot(vis_sphere, line_vec.xyz)-line_vec.w > 0.0;
}
\end{lstlisting}

\subsection[Simple particle rasterization]{Simple procedural particle rasterization}
\label{sub:particle}

\begin{figure}[H]
	\centering\includegraphics{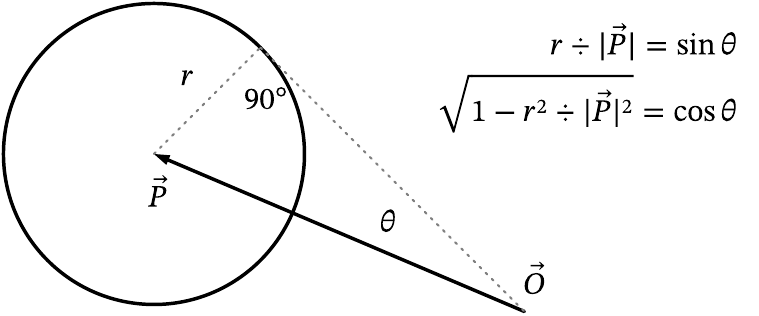}
	\caption[Simple particle model]
		{Simple particle model, where $r$ is the particle radius and $\vec P$ is the particle position from observation point $\vec O$.}
	\label{fig:simple particle}
\end{figure}
\noindent Following algorithm will produce mask image of spherical particle, given position and radius.

\begin{equation}
	\text{particle}(\hat G,\vec P,r)
		=
		\text{step}
		\left(
			\hat G \cdot \hat P - \sqrt{1-r^2 \div (\vec P\cdot\vec P)}
		\right)
\end{equation}
Where $\vec P\in\mathbb{R}^3$ is the particle position. $r$ is particle radius and $\hat G$ is the perspective map vector.
To obtain texture coordinates within the particle, following algorithm can be used.

\begin{subequations}
\begin{align}
	\begin{split}
		\hat X &= \Vert \vec P_z, 0, -\vec P_x \Vert \\
		\hat Y &= \Vert \hat X \times \vec P \Vert \equiv \hat X \times \hat P
	\end{split} \\
	\begin{bmatrix}
		\vec f_s^P \\
		\vec f_t^P
	\end{bmatrix}
	&=
	\begin{bmatrix}
		\hat G_x \\
		\hat G_y \\
		\hat G_z
	\end{bmatrix}
	\begin{bmatrix}
		\hat X_x & \hat X_y & \hat X_z \\
		\hat Y_x & \hat Y_y & \hat Y_z
	\end{bmatrix}
	\frac{|\vec P|}{2r}+\frac{1}{2} \qed
\end{align}
\end{subequations}
Where $\vec f^P$ is the texture coordinate of the particle, $\hat X$ and $\hat Y$ are rotation matrix vectors.
Full particle rasterization process with texture coordinates and round mask can be described by following algorithm.

\begin{subequations}
\begin{align}
	\begin{bmatrix}
		\vec G'_1 \\
		\vec G'_2 \\
		\vec G'_3
	\end{bmatrix}
	&=
	\begin{bmatrix}
		\hat G_x \\
		\hat G_y \\
		\hat G_z
	\end{bmatrix}
	\begin{bmatrix}
		\Vert\begin{bmatrix}
			\vec P_z & 0 & -\vec P_x
		\end{bmatrix}\Vert
		\\
		\Vert\begin{bmatrix}
			\vec P_z & 0 & -\vec P_x
		\end{bmatrix}\Vert
		\times \hat P \\
		\hat P_x \quad \hat P_y \quad \hat P_z \\
	\end{bmatrix} \\
	m^P &=
	\text{step}
	\left(
		\vec G'_3 - \sqrt{1-r^2 \div (\vec P\cdot\vec P)}
	\right) \qed \\
	\begin{bmatrix}
		\vec f_s^P \\
		\vec f_t^P
	\end{bmatrix}
	&=
	\begin{bmatrix}
		\vec G'_1 \\
		\vec G'_2
	\end{bmatrix}
	\frac{|\vec P|}{2r}+\frac{1}{2} \qed
\end{align}
\end{subequations}
Where $\vec f^P\in[0,1]^2$ are the texture coordinates, $m^P\in[0,1]$ is the particle mask, $\vec P$ is the particle position with $r$ as radius. $\hat G$ is the perspective-map vector.
% Particle rasterization functions
\begin{lstlisting}[float=p,
	label={lst:particle mask},
	caption={
		[Particle rasterization]
		Particle rasterization function in GLSL. Function $\textbit{pxstep}()$ can be found in listing \vref{lst:pixel step}.
	}
]
float dot(vec3 vector) { return dot(vector, vector); }
float sq(float scalar) { return scalar*scalar; }
vec3 normal(vec3 A, vec3 B)
{ return normalize(cross(A, B) ); }

// Once per particle's pixel
float getParticleMask(vec3 visual_sphere, vec4 particle)
{
	return pxstep(dot(visual_sphere, particle.xyz)
		-sqrt(1.0-sq(particle.w)/dot(particle.xyz) ) );
}

// Once per particle, per frame
mat2x3 getPatricleTexCoordMatrix(vec4 particle)
{
	vec3 X = normalize(vec3(particle.z, 0.0, -particle.x) );
	return mat2x3(X, normal(X, particle.xyz) );
}

// Once per particle's pixel
vec2 getParticleTexCoord(vec3 visual_sphere, mat2x3 particle_mat, vec4 particle)
{ return (visual_sphere*particle_mat)*length(particle.xyz)/particle.w*0.5+0.5; }
\end{lstlisting}

\subsection[Hidden surface problem]{Hidden surface occlusion}
\label{sub:occlusion}

One of the key features of 3D polygon rasterization is hidden surface occlusion (HSO). There are many approaches that solve HSO problem. One of them is depth-pass test, where each depth $r$ of the fragment pixel is tested against already rasterized geometry depth $r^\bot$. This simple technique works well with the binary step function (see sub-subsection \vref{sub:step function}). Depth pass test technique produces aliased result with some pixels being drawn multiple times during a single pass.

To utilize advantage of aliasing-free rasterization in perspective map rasterization, geometry data should be rasterized in a very organized order. Optimal approach is to render scene front-to-back and discard already painted pixels during polygon mask rasterization. Such ordered rasterization can be achieved with binary space partitioning\supercite{Fuchs1980BinarySpacePartitioning,Newell1972hidden_surface} (BSP) and concurrent binary trees\supercite{Dupuy2020ConcurrentBinaryTrees} (CBT). This solution produces aliasing-free result, without depth-buffer check. It paints pixels only once for most of the time.

\begin{equation}
	\begin{cases}
		m^{fc} = \min\{m^f,\ 1-m^b\} \\ % mask fragment clip
		m = m^b+m^{fc} % geometry mask
	\end{cases}
\end{equation}
Above equation combines fragment data with the buffer mask in front-to-back rasterization. Where $m^f$ is the currently processed fragment black-and-white mask. $m^b$ is the buffer mask of already rasterized geometry. Mask $m^{fc}$ is the fragment mask clipped (occluded) by the buffer mask $m^b$. Simple addition combines current fragment mask $m^{fc}$ with the buffer $m^b$.
Same addition can be performed for any fragment data, having the $m^{fc}$ as a mask. For example \emph{UV} coordinates for surface textures (equation \ref{eq:mask uv}), normal pass (equation \ref{eq:mask normal}), depth pass (equation \ref{eq:mask depth}), etc. Initial value of each buffer is equal zero.
\begin{subequations}
\begin{align}
	r &= m^{fc} r^f+r^b \label{eq:mask depth}, & r^b(1) &= 0 \\
	\vec f &= m^{fc}\vec f^f+\vec f^b \label{eq:mask uv}, & \vec f^b(1) &= \begin{bmatrix} 0 & 0 \end{bmatrix}\transpose \\
	\hat N &= \Vert m^{fc}\hat N^f+\hat N^b\Vert \label{eq:mask normal}, & \vec N^b(1) &= \begin{bmatrix} 0 & 0 & 0 \end{bmatrix}\transpose
\end{align}
\end{subequations}
Where $r$ is the depth pass, $\vec f\in[0,1]^2$ represents texture coordinates and $\hat N\in[0,1]^3$ is the normal vector. Current fragment mask (occlusion clipped) is denoted by $m^{fc}$. \textquote{$f$} in a superscript indicates current fragment data and \textquote{$b$} in superscript--buffer data of already rasterized geometry. Initial values of the buffer are denoted by \enquote{$(1)$} in suffix.

% Fragment masking and buffer merge
\begin{lstlisting}[float=p,
	label={lst:fragment occlusion},
	caption={
		[Fragment data occlusion]
		Fragment data occlusion and buffer merging, for front-to-back rasterization, in GLSL.
	}
]
// Returns fragment mask occluded by render buffer mask
float clipFragment(float fragMask, float maskBuffer)
{ return min(fragMask, 1.0-maskBuffer); }

// Combines current fragment mask with the buffer
float combineMaskPass(float fragClip, float maskBuffer)
{ return fragClip+maskBuffer; }

// Combines current fragment depth with the buffer
float combinePass(float fragClip, float fragDepth, float depthBuffer)
{ return fragClip*fragDepth+depthBuffer; }

// Combines current fragment texture coordinates with the buffer
vec2 combinePass(float fragClip, vec2 fragCoord, vec2 texBuffer)
{ return fragClip*fragCoord+texBuffer; }

// Combines current fragment normal with the buffer
vec3 combinePass(float fragClip, vec3 fragNormal, vec3 normBuffer)
{ return fragClip*fragNormal+normBuffer; }
\end{lstlisting}

\section[Perspective picture transformations]{Perspective transformations of 2D/3D data}
\label{sec:transformations}

In this section presented algorithms produce perspective picture from 3D and 2D data.
Most of them describe 2D$\rightarrow$3D transformation, which outputs perspective vector map from texture coordinates, suitable for rasterization.

\begin{rem}
	For a proper transformation, 2D coordinates must be normalized for a given AOV type (e.g. vertical, diagonal or horizontal).
	\begin{example*}
		For a pixel $i$ in picture of aspect-ratio 16:9 and horizontal-AOV, coordinates $(i_x,i_y)$ must be centered and horizontally normalized, so that $i_x\in[-1,1]$ and $i_y\in\left[-\frac{9}{16},\frac{9}{16}\right]$.

		\begin{equation}
			\begin{bmatrix}
				\vec f_x \\
				\vec f_y
			\end{bmatrix}
			=
			\left(
				2\begin{bmatrix}
					\vec f_s \\
					\vec f_t
				\end{bmatrix}-1
			\right)
			\begin{cases}
				\begin{bmatrix}1 & \frac{1}{a}\end{bmatrix}\transpose
					,& \text{if $\Omega$ horizontal}
					\smallskip\\
				\begin{bmatrix}a & 1\end{bmatrix}\transpose
					,& \text{if $\Omega$ vertical}
					\smallskip\\
				\begin{bmatrix}a & 1\end{bmatrix}\transpose \div \sqrt{1+a^2}
					,& \text{if $\Omega$ diagonal}
					\smallskip\\
				\begin{bmatrix}a & 1\end{bmatrix}\transpose \div \frac{4}{3}
					,& \text{if $\Omega$ horizontal 4$\times$3}
			\end{cases}
		\end{equation}
		Above are view coordinates $x,y\in[-1,1]$ from texture coordinates $s,t\in[0,1]$, where
		$a$ is the picture aspect ratio and $\Omega$ is the AOV.

		\begin{equation}
			\begin{bmatrix}
				\vec f_s \\
				\vec f_t
			\end{bmatrix}
			=
			\frac{1}{2}+\frac{1}{2}
			\begin{bmatrix}
				\vec f_x \\
				\vec f_y
			\end{bmatrix}
			\begin{cases}
				\begin{bmatrix}1 & a\end{bmatrix}\transpose
					,& \text{if $\Omega$ horizontal}
					\smallskip\\
				\begin{bmatrix}\frac{1}{a} & 1\end{bmatrix}\transpose
					,& \text{if $\Omega$ vertical}
					\smallskip\\
				\begin{bmatrix}\frac{1}{a} & 1\end{bmatrix}\transpose \sqrt{1+a^2}
					,& \text{if $\Omega$ diagonal}
					\smallskip\\
				\begin{bmatrix}\frac{1}{a} & 1\end{bmatrix}\transpose\frac{4}{3}
					,& \text{if $\Omega$ horizontal 4$\times$3}
			\end{cases}
		\end{equation}
		Here are texture coordinates $s,t\in[0,1]$ from view coordinates $x,y\in[-1,1]$, where
		$a$ is the picture aspect ratio and $\Omega$ is the AOV.
	\end{example*}
\end{rem}

% Texture coordinates transformations
\begin{lstlisting}[float=p,
	label={lst:coordinates},
	caption={
		[Screen coordinates transformation]
		Texture to screen coordinates transformation with aspect ratio and FOV conversion, in GLSL.
	}
]
// Convert texture coordinates to screen coordinates
vec2 getScreenCoord(vec2 texCoord, float aspectRatio, int fovType)
{
	// FOV type:
	// 0 - horizontal
	// 1 - diagonal
	// 2 - vertical
	// 3 - horizontal 4x3

	// Center coordinates
	vec2 scrCoord = texCoord*2.0-1.0;
	// Correct aspect ratio
	if (fovType==0) // Horizontal
		scrCoord.y /= aspectRatio;
	else // Vertical, diagonal, or 4x3
	{
		scrCoord.x *= aspectRatio;
		// Diagonal or 4x3
		if (fovType == 1) // Diagonal
			scrCoord *= inversesqrt(aspectRatio*aspectRatio+1.0);
		else if (fovType == 3) // 4x3
			scrCoord /= 4.0/3.0;
	}

	return scrCoord;
}

// Convert screen coordinates to texture coordinates
vec2 getTextureCoord(vec2 scrCoord, float aspectRatio, int fovType)
{
	// FOV type:
	// 0 - horizontal
	// 1 - diagonal
	// 2 - vertical
	// 3 - horizontal 4x3

	// Map aspect ratio to square
	if (fovType==0) // Horizontal
		scrCoord.y *= aspectRatio;
	else // Vertical, diagonal, or 4x3
	{
		scrCoord.x /= aspectRatio;
		// Diagonal or 4x3
		if (fovType == 1) // Diagonal
			scrCoord *= sqrt(aspectRatio*aspectRatio+1.0);
		else if (fovType == 3) // 4x3
			scrCoord *= 4.0/3.0;
	}

	// Map to corner
	return scrCoord*0.5+0.5;
}
\end{lstlisting}

\subsection[Universal perspective]{Universal perspective model}
\label{sub:universal}

Universal perspective model allows for a smooth adjustment of image geometry in accordance to the visible content.
Presented two transforms produce \emph{perspective picture} (see definition \vpageref{def:azimuthal perspective}).
First 2D$\rightarrow$3D transformation produces perspective vector in various common projections, which is suitable for generating universal-perspective maps for rasterization.
Second 3D$\rightarrow$2D transformation produces picture coordinates of various projections.
Combination of both can map between different projections.
\begin{rem}
	Note that 3D$\rightarrow$2D transformation may be non-linear, meaning that a 3D vector can be represented by multiple 2D coordinates in very special cases.
	\begin{example}\label{ex:nonlinear mapping}
		In \emph{equidistant} projection at $\Omega=2\pi$ (whole sphere) visible point opposite to the view direction is represented by a ring at the picture's boundary.
	\end{example}
\end{rem}

\subsubsection[Parameters and limits]{Universal perspective parameters and limits}
\label{sub:perspective parameters}

There is total of four parameters defining perspective picture geometry in the universal perspective model.
\begin{description}
	\item [angle $\Omega$]
		defines angle of view (AKA FOV), which varies between $\rightarrow\pi$ and $2\pi$.
	\item [scalar $k$]
		defines perspective type as a value in range $[-1,1]$, which interpolates between various azimuthal projections:
		\begin{center}
			\begin{tabular}{c@{\quad$k=$ }d}
				\emph{Gnomonic (rectilinear)} & 1 \\
				\emph{Stereographic} & 0.5 \\
				\emph{Equidistant} & 0 \\
				\emph{Equisolid} & -0.5 \\
				\emph{Orthographic} & -1 \\
			\end{tabular}
		\end{center}
	\item [scalar $l$]
		in range $[0,1]$ defines cylindricity of the projection, where $l=0$ represents cylindrical projection and $l=1$ --- spherical projection.
	\item [scalar $s$]
		in range $[\nicefrac{4}{5},1]$ defines vertical anamorphic correction of a non-spherical projection, meaning that it is only active when $l<1$.
\end{description}
Base projection type is adjusted by the $k$ component. It manipulates image perception. Cylindrical projection, is adjusted by the $l$ component. Low $l$ values should represent view at level (see subfigure \vref{fig:Pitch-yaw-motion}). For roll motion, recommended value for $l$ is $100\%$ (see subfigure \ref{fig:Roll-motion}). Anamorphic correction of non-spherical image, driven by the $s$ component, depends on subject in view. Purpose of the $s$ scalar is to adjust proportions of the figure-in-focus.

\begin{equation}
	\begin{cases}
		\begin{aligned}
			% Maximum FOV
			\Omega_{\max}
			&=
			\frac
			{1}
			{\max\big\{\nicefrac{1}{2}, |k|\big\}} \cdot
			\begin{cases}
				\rightarrow\pi % approaching 180 degrees
				,& \text{if } k>0 \\
				\pi,& \text{otherwise}
			\end{cases} \\
			\Omega &\in \big(0, \Omega_{\max}\big] \\
			% Distortion parameter
			k &\in [-1,1] \\
			% Cylindrical
			l &\in [0,1] \\
			% Anamorphic
			s &\in [\nicefrac{4}{5},1]
		\end{aligned}
	\end{cases}
\end{equation}
Limits for perspective parameters in universal perspective model. Value of $\Omega_{\max}$ depends on the projection type (represented by scalar $k$) and varies between angle approaching 180\degree\ and 360\degree\ angle.

% limits for perspective parameters
\begin{lstlisting}[float=p,
	label={lst:limits},
	caption={
		[Perspective parameters limiting]
		Function for range limiting of universal perspective parameters ($k$, $l$, $s$ and $\Omega$) in GLSL.
	}
]
void limits(inout float k, inout float l, inout float s, inout float fov)
{
	k = clamp(k,-1.0, 1.0);
	l = clamp(l, 0.0, 1.0);
	s = clamp(s, 0.8, 1.0);
	// Field of view in degrees
	float fov_max = (k>0.0? 179.0:180.0) / max(abs(k), 0.5);
	fov = clamp(fov, 1.0, fov_max);
}
\end{lstlisting}

\subsubsection[2D--3D transformation]{Transformation of 2D$\rightarrow$3D coordinates}
\label{sub:2D-3D transform}

This transformation produces visual sphere vector map from texture coordinates, that can be later used as an input for perspective map rasterizer.

\begin{subequations}
\begin{align}
	R &= \left\vert
		\begin{bmatrix}
			\vec f_x \smallskip\\
			\vec f_y \cdot \sqrt l
		\end{bmatrix}
	\right\vert
\\
	&\equiv \sqrt{\vec f_x^2+l\vec f_y^2}
\\
	\theta &=
	\begin{cases}
		\arctan{ \big( \tan{ \left( k\frac{\Omega}{2} \right) }R \big) } \div k, & \text{if } k>0 \\
		\frac{\Omega}{2} R, & \text{if } k=0 \\
		\arcsin{ \big( \sin{ \left( k\frac{\Omega}{2} \right) }R \big) } \div k, & \text{if } k<0
	\end{cases}
\\
	\begin{bmatrix}
		\hat v_x \\
		\hat v_y \\
		\hat v_z
	\end{bmatrix}
	&=
	\left\Vert
		\begin{bmatrix}
			\vec f_x \\
			\vec f_y \\
			1
		\end{bmatrix}
		\begin{bmatrix}
			\nicefrac{\sin(\theta)}{R} \\
			\nicefrac{\sin(\theta)}{R} \\
			\cos\theta
		\end{bmatrix}
		\div
		\begin{bmatrix}
			1 \\
			l(1-s)+s \\
			1
		\end{bmatrix}
	\right\Vert \qed
\end{align}
\end{subequations}
Picture coordinates are denoted by vector $\vec f\in[0,1]^2$, $a$ is picture aspect ratio. Transformed 3D coordinates are represented by normalized vector $\hat v\in[-1,1]^3$.
Scalar $k$ represents various projection types.
Scalar $l\in[0,1]$ is the spherical projection factor, with $l=0$ representing cylindrical projection and $l=1$, a spherical projection. Scalar $s\in[\nicefrac{4}{5},1]$ describes anamorphic correction of non-spherical image. For $s=1$ or $l=1$ there is no anamorphic correction.
% Universal perspective map generator
\begin{lstlisting}[float=p,
	label={lst:universal 2D-3D},
	caption={
		[Universal perspective map]
		Visual sphere vector $\hat v\in{[-1,1]^3}$ function from texture coordinates $\vec f\in[0,1]^2$ in GLSL, for universal-perspective system.
	}
]
vec2 getViewCoord(vec2 tex_coord, float aspect, int fov_type)
{
	tex_coord.xy = 2.0*tex_coord.st-1.0;
	switch (fov_type)
	{
		default: // horizontal FOV (type 1)
			tex_coord.y /= aspect;
			break;
		case 2: // diagonal FOV (type 2)
			tex_coord.xy /= length(vec2(aspect, 1.0) );
		case 3: // vertical FOV (type 3) and type 2 final step
			tex_coord.x *= aspect;
			break;
		case 4: // 4x3 horizontal FOV (type 4)
			tex_coord.xy *= vec2(0.75*aspect, 0.75);
			break;
	}
	return tex_coord;
}

vec3 getPerspectiveMap(vec2 view_coord, float fov, float k, float l, float s)
{
	float halfOmega = radians(fov*0.5);
	float R = length(vec2(view_coord.x, l*view_coord.y) );

	float theta;
	if (k>0.0) theta = atan(tan(k*halfOmega)*R)/k;
	else if (k<0.0) theta = asin(sin(k*halfOmega)*R)/k;
	else theta = halfOmega*R;

	view_coord.y /= l*(1.0-s)+s;
	return normalize(
		vec3(view_coord*sin(theta)/R, cos(theta) )
	);
}
\end{lstlisting}

\subsubsection[3D--2D transformation]{Transformation of 3D$\rightarrow$2D coordinates}
\label{sub:3D-2D transform}

This transformation is mainly used in a pixel shader, where basic rectilinear projection can be mapped to spherical one. It is not suitable for mapping 3D points onto 2D picture plane as in some specific cases single 3D point can map to multiple 2D positions (see example \vref{ex:nonlinear mapping}).

\begin{subequations}
\begin{align}
	\hat v &= \begin{bmatrix} \hat v_x \\ \hat v_y \\ \hat v_z \end{bmatrix}
\\
	\theta &= \arccos{ \left( \hat v_z \div \sqrt{ \hat v_x^2+l\hat v_y^2+\hat v_z^2 } \right) }
\\
	R &=
	\begin{cases}
		\tan{ \left( k\theta\right ) } \div \tan{ \left( k\frac{\Omega}{2} \right) }, & \text{if } k>0 \\
		\theta\div\frac{\Omega}{2}, & \text{if } k=0 \\
		\sin{ \left( k\theta\right ) } \div \sin{ \left( k\frac{\Omega}{2} \right) }, & \text{if } k<0
	\end{cases}
\\
	\begin{bmatrix}
		\vec f_x \\
		\vec f_y
	\end{bmatrix}
	&=
	\begin{bmatrix}
		\hat v_x \\
		\hat v_y
	\end{bmatrix}
	\frac{R}{ \sqrt{ \hat v_x^2 +l \hat v_y^2 } }
	\begin{bmatrix}
		1 \\
		l\left(1-s\right)+s
	\end{bmatrix} \qed
\end{align}
\end{subequations}
3D coordinates are represented by a normalized vector $\hat v\in[-1,1]^3$, where view origin is at position $[0\quad0\quad0]$. Transformed picture coordinates are represented by vector $\vec f\in\real^2$, where image center is at position $[0\quad0]$. Angle $\theta$ is between vector $\hat v$ and the $Z$ axis. $R$ is the normalized distance between projected vector $\vec f$ and the image center, where $\vec f\mapsfrom\hat v$. Angle $\Omega$ is equal to AOV, where $\Omega_{\max}\in[\rightarrow\pi,2\pi]$.
Scalar $k$ represents various projection types (see sub-subsection \vref{sub:perspective parameters}).
Scalar $l\in[0,1]$ is the spherical projection factor, where $l=0$ represents cylindrical projection and $l=1$ represents spherical projection. Scalar $s\in[\nicefrac{4}{5},1]$ describes anamorphic correction of non-spherical image. For $s=1$ or $l=1$ there is no anamorphic correction.
This transformation is a reverse of the universal 2D$\rightarrow$3D transform \vpageref{sub:2D-3D transform}.

\subsubsection[2D--2D transformation]{Transformation of 2D$\rightarrow$2D coordinates}
\label{sub:2D-2D transform}

Combination of 3D and 2D transformation can be used to map between two different projections, for example $\text{\emph{stereographic}}\leftrightarrow\text{\emph{equidistant}}$, using two separate perspective component sets, $\{\Omega_i,k_i,l_i,s_i\}$ and $\{\Omega_o,k_o,l_o,s_o\}$ for input and output picture.

\begin{subequations}
\begin{align}
	R_i &= \left\vert
		\begin{bmatrix}
			\vec f_x \smallskip\\
			\vec f_y \cdot \sqrt{l_i}
		\end{bmatrix}
	\right\vert
\\
	&\equiv \sqrt{\vec f_x^2+l_i\vec f_y^2}
\\
	\theta &=
	\begin{cases}
		\arctan{ \big( \tan{ \left( k_i\frac{\Omega_i}{2} \right) }R_i \big) } \div k_i, & \text{if } k_i>0 \\
		\frac{\Omega_i}{2} R_i, & \text{if } k_i=0 \\
		\arcsin{ \big( \sin{ \left( k_i\frac{\Omega_i}{2} \right) }R_i \big) } \div k_i, & \text{if } k_i<0
	\end{cases}
\\
	\begin{bmatrix}
		\hat v_x \\
		\hat v_y \\
		\hat v_z
	\end{bmatrix}
	&=
	\left\Vert
		\begin{bmatrix}
			\vec f_x \\
			\vec f_y \\
			1
		\end{bmatrix}
		\begin{bmatrix}
			\nicefrac{\sin(\theta)}{R_i} \\
			\nicefrac{\sin(\theta)}{R_i} \\
			\cos\theta
		\end{bmatrix}
		\div
		\begin{bmatrix}
			1 \\
			l_i(1-s_i)+s_i \\
			1
		\end{bmatrix}
	\right\Vert
\\
	R_o &=
	\begin{cases}
		\tan{ \left( k_o\theta\right ) } \div \tan{ \left( k_o\frac{\Omega_o}{2} \right) }, & \text{if } k_o>0 \\
		\theta\div\frac{\Omega_o}{2}, & \text{if } k_o=0 \\
		\sin{ \left( k_o\theta\right ) } \div \sin{ \left( k_o\frac{\Omega_o}{2} \right) }, & \text{if } k_o<0
	\end{cases}
\\
	\begin{bmatrix}
		\vec f'_x \\
		\vec f'_y
	\end{bmatrix}
	&=
	\begin{bmatrix}
		\hat v_x \\
		\hat v_y
	\end{bmatrix}
	\frac{R_o}{ \sqrt{ \hat v_x^2 +l_o \hat v_y^2 } }
	\begin{bmatrix}
		1 \\
		l_o\left(1-s_o\right)+s_o
	\end{bmatrix} \qed
\end{align}
\end{subequations}
Input picture coordinates are represented by $\vec f\in\real^2$, while output picture coordinates are denoted as $\vec f'\in\real^2$.

\subsection[Various projections]{Perspective map algorithms for various projections}
\label{sub:various projections}

Algorithms presented below produce visual sphere vector map in various projections, which can be later used as an input for perspective map rasterizer. Result vector $\hat v\in[-1,1]^3$ can be mapped to picture color range $[0,1]^3$ by simple transformation:
\begin{equation}
	\vec G'=\frac{\hat v+1}{2}
\end{equation}

\subsubsection[Rectilinear perspective]{Rectilinear perspective map}
\label{sub:linear perspective}

\begin{equation}
	\begin{bmatrix}
		\hat v_x \\
		\hat v_y \\
		\hat v_z
	\end{bmatrix}
	=
	\begin{cases}
		\left\Vert
		\begin{matrix}
			2\vec f_s-1 \\
			(2\vec f_t-1)\div a \\
			\cot\frac{\Omega^h}{2}
		\end{matrix}
		\right\Vert, & \text{if $\Omega$ horizontal}
	\smallskip \\
		\left\Vert
		\begin{matrix}
			a(2\vec f_s-1)\div\sqrt{a^2+1} \\
			(2\vec f_t-1)\div\sqrt{a^2+1} \\
			\cot\frac{\Omega^d}{2}
		\end{matrix}
		\right\Vert, & \text{if $\Omega$ diagonal}
	\smallskip \\
		\left\Vert
		\begin{matrix}
			a(2\vec f_s-1) \\
			2\vec f_t-1 \\
			\cot\frac{\Omega^v}{2}
		\end{matrix}
		\right\Vert, & \text{if $\Omega$ vertical}
	\end{cases}
\end{equation}
Linear perspective map formula, where $\hat v$ is the visual sphere vector, $\vec f$ represents screen coordinates. $a$ is the screen aspect ratio and $\Omega$ the AOV.
% Rectilinear projection visual sphere vector
\begin{lstlisting}[float=p,
	label={lst:rectilinear},
	caption={
		[Rectilinear map]
		Rectilinear perspective visual-sphere vector $\hat v\in{[-1,1]^3}$ from texture coordinates $\vec f\in[0,1]^2$ in GLSL.
	}
]
vec3 getRectilinearMap(vec2 tex_coord, float aspect, float fov, int fov_type)
{
	vec3 direction = vec3(
		2.0*tex_coord-1.0, // map to range [-1,1]
		1.0/tan(0.5*radians(fov) )
	);
	switch (fov_type)
	{
		default: // horizontal FOV (type 1)
			direction.y /= aspect;
			break;
		case 2: // diagonal FOV (type 2)
			direction.xy /= length(vec2(aspect, 1.0) );
		case 3: // vertical FOV (type 3) and type 2 final step
			direction.x *= aspect;
			break;
	}

	return normalize(direction);
}
\end{lstlisting}

\subsubsection[Curved panorama]{Curved panorama map}
\label{sub:curved panorama}

\begin{figure}[H]
	\centering\includegraphics{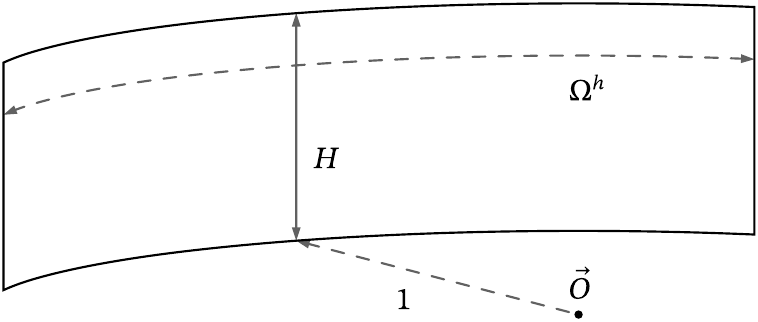}
	\caption[Curved panorama]
		{Curved panorama model, where aspect ratio $a=\frac{\Omega^h}{h}$, radius is equal to 1 and $\vec O$ denotes arc's origin.}
	\label{fig:curved panorama}
\end{figure}

\begin{subequations}
\begin{align}
	a&=\frac{\Omega^h}{H} \qed \\
	\begin{bmatrix}
		\vec f'_x \\
		\vec f'_y
	\end{bmatrix}
	&=
	\begin{bmatrix}
		\Omega^h \\
		H
	\end{bmatrix}
	\left(
	\begin{bmatrix}
		\vec f_s \\
		\vec f_t
	\end{bmatrix} -\frac{1}{2}
	\right) \\
	\begin{bmatrix}
		\hat v_x \\
		\hat v_y \\
		\hat v_z
	\end{bmatrix}
	&=
	\left\Vert\begin{matrix}
		\sin\vec f'_x \\
		\vec f'_y \\
		\cos\vec f'_x
	\end{matrix}\right\Vert \qed
\end{align}
\end{subequations}
Curved panorama perspective-map formula, where $a$ is the panorama aspect ratio, $H$ in the display height-to-radius proportion with $\Omega^h$ as the horizontal AOV. Vector $\vec f$ represents screen coordinates and $\hat v$ is the visual sphere vector.
% Panorama shader
\begin{lstlisting}[float=p,
	label={lst:panorama},
	caption={
		[Curved panorama map]
		Curved panorama perspective visual-sphere vector $\hat v\in{[-1,1]^3}$ from texture coordinates $\vec f\in[0,1]^2$ in GLSL.
	}
]
float getPanormaAspect(float horizontal_fov, float height)
{ return radians(horizontal_fov)/height; }

vec3 getPanoramaMap(vec2 tex_coord, float horizontal_fov, float height)
{
	tex_coord.xy = vec2(horizontal_fov, height)*(tex_coord.st-0.5);

	return normalize(vec3(
		sin(tex_coord.x),
		tex_coord.y,
		cos(tex_coord.x)
	));
}
\end{lstlisting}

\subsubsection[Full dome]{Full dome map}
\label{sub:full dome}

\begin{subequations}
\begin{align}
	% Centered coordinates
	\begin{bmatrix}
		\vec f'_x \\
		\vec f'_y
	\end{bmatrix}
	&=
	\begin{bmatrix}
		2\vec f_s-1 \\
		1-2\vec f_t
	\end{bmatrix} \\
	% Radial angle
	\theta &= |\vec f'|\left(\Omega+\frac{\pi}{2}\right) \\
	% Visual sphere vector
	\begin{bmatrix}
		\hat v_x \\
		\hat v_y \\
		\hat v_z
	\end{bmatrix}
	&=
	\left\Vert\begin{matrix}
		\vec f'_x\sin\theta \div |\vec f'| \\
		\cos\theta \\
		\vec f'_y\sin\theta \div |\vec f'| +o
	\end{matrix}\right\Vert
	% Rotation matrix
	\begin{bmatrix}
		1	&0				&0 \\
		0	&\cos\varphi	&-\sin\varphi \\
		0	&\sin\varphi	&\cos\varphi
	\end{bmatrix}  \qed \\
	% Radial mask
	m &= \left\{\frac{1-|\vec f'|}{\partial(1-|\vec f'|)}\right\}\cap[0,1] \qed
\end{align}
\end{subequations}
Full dome perspective-map formula, where $\vec f$ is the screen coordinates vector. $\Omega$ is the compression angle, $\varphi$ is the tilt angle. $\hat v$ represents visual sphere vector. $o$ is the view position offset (in radius) and $m$ is the radial mask with $\partial(x)$ being the equivalent of $\textbit{fwidth}(x)$ function.
% Full dome generator
\begin{lstlisting}[float=p,
	label={lst:dome},
	caption={
		[Full dome map]
		Full dome visual-sphere vector $\hat v\in{[-1,1]^3}$ from texture coordinates $\vec f\in[0,1]^2$ in GLSL.
	}
]
vec4 getDomeMap(vec2 tex_coord, float compression, float tilt, float offset)
{
	tex_coord.xy = vec2(
		2.0*tex_coord.s-1.0,
		1.0-2.0*tex_coord.t
	);
	float R = length(tex_coord);
	float theta = R*radians(compression+90.0);

	vec4 perspective_map = vec4(
		tex_coord*sin(theta)/R,
		cos(theta),
		1.0-R
	).xzyw;
	perspective_map.a = clamp(perspective_map.a
		/fwidth(perspective_map.a), 0.0, 1.0);

	perspective_map.z += offset;
	perspective_map.xyz = normalize(perspective_map.xyz);

	tilt = radians(tilt);
	vec2 rotation = vec2(sin(tilt), cos(tilt) );
	perspective_map.xyz *= mat3(
		vec3(1.0,         0.0,          0.0),
		vec3(0.0, rotation[1], -rotation[0]),
		vec3(0.0, rotation[0],  rotation[1])
	);

	return perspective_map;
}
\end{lstlisting}

\subsubsection[Equirectangular projection]{Equirectangular projection map}
\label{sub:equirectangular}

\begin{subequations}
\begin{align}
	% Screen vector
	\begin{bmatrix}
		\vec f'_x \\
		\vec f'_y
	\end{bmatrix}
	&=
	\pi
	\begin{bmatrix}
		2\vec f_s-1 \\
		\vec f_t
	\end{bmatrix} \\
	% Visual sphere vector
	\begin{bmatrix}
		\hat v_x \\
		\hat v_y \\
		\hat v_z
	\end{bmatrix}
	&=
	\begin{bmatrix}
		\sin\vec f'_x \\
		-1 \\
		\cos\vec f'_x
	\end{bmatrix}
	\begin{bmatrix}
		\sin\vec f'_y \\
		\cos\vec f'_y \\
		\sin\vec f'_y
	\end{bmatrix} \qed
\end{align}
\end{subequations}
Equirectangular projection perspective-map formula, where $\vec f$ represents screen coordinates and $\hat v$ is the visual sphere vector.
% Equirectangular map generator
\begin{lstlisting}[float=p,
	label={lst:equirectangular},
	caption={
		[Equirectangular map]
		Equirectangular projection visual-sphere vector $\hat v\in{[-1,1]^3}$ from texture coordinates $\vec f\in[0,1]^2$ in GLSL.
	}
]
vec3 getEquirectangularMap(vec2 tex_coord)
{
	tex_coord.x = 2.0*tex_coord.s-1.0;
	tex_coord.xy *= radians(180.0);

	vec2 sine = sin(tex_coord);
	vec2 cosine = cos(tex_coord);

	return vec3(sine.x*sine.y, -cosine.y, cosine.x*sine.y);
}
\end{lstlisting}

\subsubsection[Mirror dome projection]{Mirror dome projection map}
\label{sub:mirror dome}

\begin{figure}[H]
	\centering\includegraphics{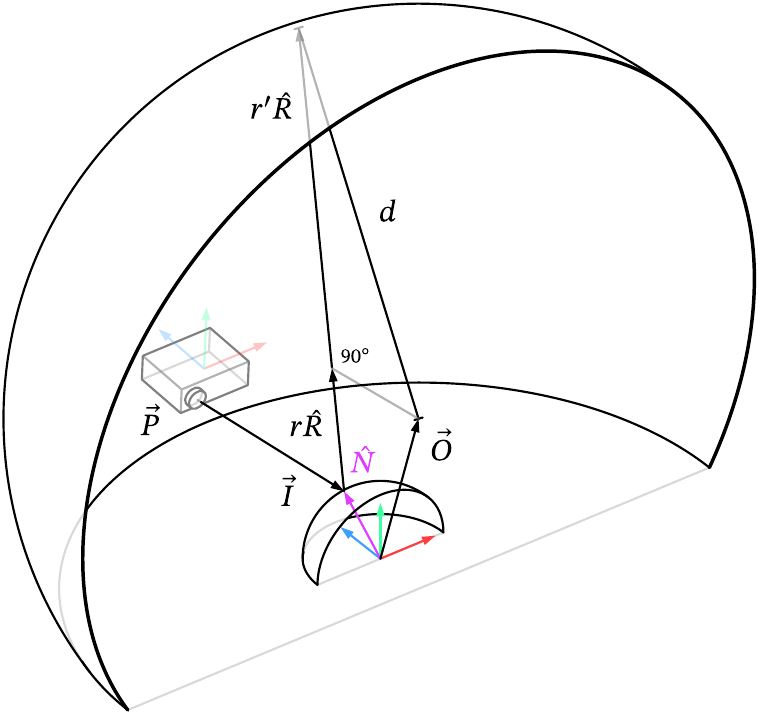}
	$\acwunderarcarrow$
	\includegraphics{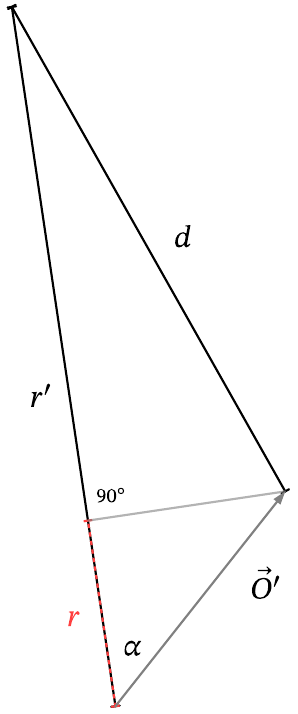}
	\caption[Mirror dome projection]
		{Mirror dome projection model, with mirror center as world origin.\supercite{Burke2008iDome}}
	\label{fig:mirror dome}
\end{figure}
\begin{subequations}
\begin{align}
	\vec I &= \hat N-\vec P \\
	\hat R &= \Vert \vec I-2(\vec I\cdot\hat N)\hat N \Vert \\
	&\hspace{-1em}\left\{
	\begin{aligned}
		r &= \hat R\cdot(\vec O-\hat N) \\
		r' &= r+\sqrt{d^2-|r\hat R+\hat N-\vec O|^2}
	\end{aligned}
	\right.
\\
	\hat v &= \frac{r'\hat R +\hat N-\vec O}{d} \qed \\
	m &= \left(\frac{r'+|\vec I|}{(r'+|\vec I|)^{\max}}\right)^2 \qed
\end{align}
\end{subequations}
Mirror dome projection perspective-map formula, where $\vec I$ represents incident vector. $\hat N\in[-1,1]^3$ is the spherical-mirror world-normal and surface position, $\vec P$ is the projector position, $\vec R$ is reflection vector, $r'$ is the reflection distance to dome intersection, as:
\begin{equation}
	\hat R\cdot(\vec O-\hat N)
	= \hat R\cdot\vec O'
	= \cancel{|\hat R|}|\vec O'|\cos\alpha
	= \cancel{|\vec O'|}\frac{r}{\cancel{|\vec O'|}}
	= r
\end{equation}
$\hat v\in[-1,1]^3$ is the visual sphere vector as mirror surface color and $\vec O$ is the dome origin position with $d$ as dome radius. Mirror radius is equal $1$ with its origin at position $[0\quad0\quad0]$. Light dimming mask is represented by $m$ and it's based on inverse-square law approximation. In order to produce perspective map image, first mirror 3D model world-normal pass must be produced, as viewed from projector's perspective. It is possible to render the view with additional perspective map of the projector.
% Mirror dome mirror color generator
\begin{lstlisting}[float=p,
	label={lst:mirror dome},
	caption={
		[Mirror dome map]
		Mirror-dome visual-sphere vector $\hat v\in{[-1,1]^3}$ from mirror normal-pass, as seen through projector in GLSL.
	}
]
float dot(vec3 vector) { return dot(vector, vector); }
float sq(float scalar) { return scalar*scalar; }
float pxstep(float scalar)
{ return clamp(scalar/fwidth(scalar), 0.0, 1.0); }

vec4 getMirrorDomeMap(vec3 mirror_normal, vec3 projector_map, mat3 projector_rot, vec3 projector_pos, float dome_radius, vec3 dome_pos)
{
	vec3 reflection = reflect(projector_map, mirror_normal)*projector_rot;
	vec3 mirror_surface = mirror_normal*projector_rot;
	dome_pos -= mirror_surface; // set current mirror surface as origin

	// Get reflection length to dome intersection
	float length = dot(reflection.xyz, dome_pos.xyz);
	length += sqrt( dome_radius*dome_radius-dot(dome_pos.xyz-length*reflection.xyz) );

	return vec4(
		(length*reflection.xyz-dome_pos.xyz)/dome_radius,
		length+length(mirror_surface.xyz-projector_pos.xyz)
	);
}

vec4 processMirrorDomeMask(vec2 tex_coord, vec3 dome_pos)
{
	vec2 pixSize = 1.0/sizeof(MirrorDomeMap);
	vec4 mirror_dome = texture(MirrorDomeMap, tex_coord);

	// Get mirror-bounds mask
	float alpha = texture(MirrorNormalPass, tex_coord).a;
	// Generate reflected-dome-bounds mask
	alpha *= pxstep(mirror_dome.z); // half dome clip
	alpha *= pxstep(mirror_dome.y+dome_pos.y); // bottom clip

	// Get maximum incident length for normalization
	float max_incident = 0.0;
	for (int y=0; y<sizeof(MirrorDomeMap).y; y++)
	for (int x=0; x<sizeof(MirrorDomeMap).x; x++)
		max_incident = max(
			max_incident,
			alpha*texture(MirrorDomeMap, vec2(x,y)*pixSize).w
		);

	// Combine normalized brightness mask
	alpha *= sq(mirror_dome.w/max_incident);

	return vec4(mirror_dome.xyz, alpha);
}
\end{lstlisting}
\begin{rem}
	Mirror dome projection system was originally developed by \noun{P. Bourke}.\supercite{Burke2008iDome}
\end{rem}

\subsubsection[Projection mapping]{Projection mapping perspective map}
\label{sub:projection mapping}

\begin{figure}[H]
	\centering\includegraphics{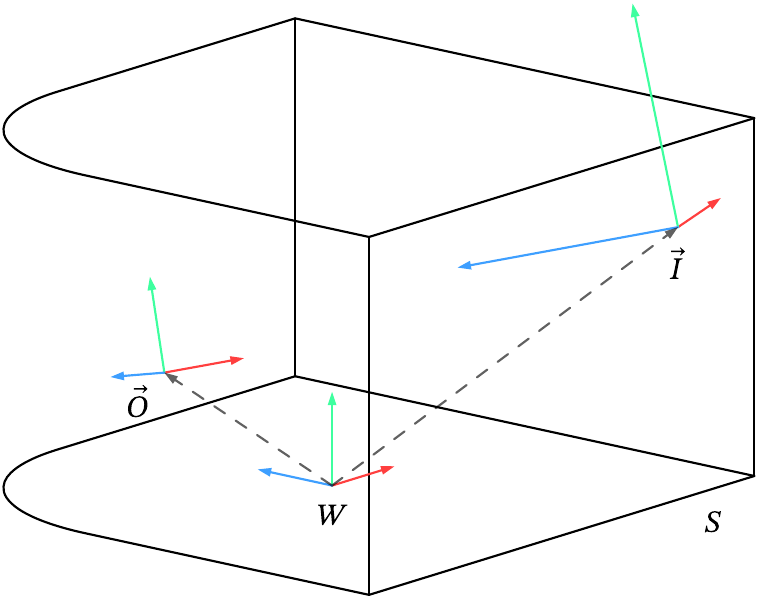}
	\caption[Projection mapping]
		{Projection mapping model, with $W$ as the world origin coordinates, $\vec O$ as observation point, $\vec I$ as projector position and $S$ as projection surface.}
	\label{fig:projection mapping}
\end{figure}

\begin{subequations}
\begin{align}
	\begin{bmatrix}
		\hat v_x \\
		\hat v_y \\
		\hat v_z
	\end{bmatrix}
	&=
	\left\Vert
	\begin{matrix}
		\vec S_x - \vec O_x \\
		\vec S_y - \vec O_y \\
		\vec S_z - \vec O_z
	\end{matrix}
	\right\Vert
	\begin{bmatrix}
		W'_{11} & W'_{12} & W'_{13} \smallskip \\
		W'_{21} & W'_{22} & W'_{23} \smallskip \\
		W'_{31} & W'_{32} & W'_{33}
	\end{bmatrix} \qed \\
	m &= \left(\frac{|\vec S-\vec I|}{|\vec S-\vec I|^{\max}}\right)^2 \qed
\end{align}
\end{subequations}
Projection mapping perspective-map formula, where $\hat v$ represents the visual sphere unit-vector. $\vec S$ is the environment surface position, as seen through projector's point of view. $\vec O$ denotes observation point. $W'$ is the optional rotation matrix for the world position $\vec S$. Light dimming mask is denoted by $m$, which value is based on the inverse-square law. $\vec I$ represents the projector world position.

In order to produce perspective map image, first 3D model of the projection environment must be produced. Then world position map $S$ must be rendered from projector point-of-view. Position map $S$ changes only if projector position $\vec I$ changes, so the map can be reused at each frame even if the observation position is dynamic.
For more about projection mapping, see example \ref{ex:projection mapping} in sub-subsection \vref{sub:perspective shader}.
% projection mapping function
\begin{lstlisting}[float=p,
	label={lst:projection mapping},
	caption={
		[Projection-mapping map]
		Visual sphere vector $\hat v\in{[-1,1]^3}$ function from world-position pass as seen through projector in GLSL, for projection mapping.
	}
]
vec3 getProjectionMap(vec2 tex_coord, vec3 view_pos, mat3 view_mat)
{
	vec3 world_pos = texture(WorldPos, tex_coord).xyz;
	return normalize(world_pos-view_pos)*view_mat;
}

float getBrightnessMask(vec2 tex_coord)
{
	vec2 pixSize = 1.0/sizeof(WorldPos);

	float max_dist2 = 0.0;
	for (int y=0; y<sizeof(WorldPos).y; y++)
	for (int x=0; x<sizeof(WorldPos).x; x++)
	{
		vec3 incident = texture(WorldPos, pixSize*vec2(x,y)).xyz
		max_dist2 = max(max_dist2, dot(incident, incident));
	}

	vec3 incident = texture(WorldPos, tex_coord).xyz;
	return dot(incident, incident)/max_dist2;
}
\end{lstlisting}

\subsubsection[Cube-mapping]{Cube-mapping perspective map}
\label{sub:cube mapping}

\begin{figure}[H]
	\centering\includegraphics{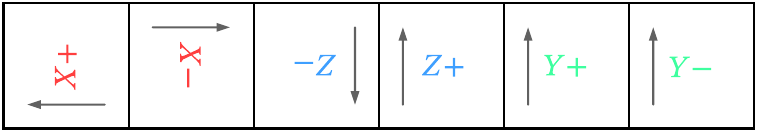}
	\caption[Cube map]
		{Cube map model with face orientation.}
	\label{fig:cubemap}
\end{figure}

\begin{subequations}
\begin{gather}
	\begin{aligned}
		&\begin{cases}
			\hat X = \left[ 1\quad 0\quad 0 \right]\transpose \\
			\hat Y = \left[ 0\quad 1\quad 0 \right]\transpose \\
			\hat Z = \left[ 0\quad 0\quad 1 \right]\transpose
		\end{cases} \\
		&\hspace{1.1em} i = \left\lfloor 6 \vec f_s \right\rfloor
	\end{aligned}
\\
	\begin{bmatrix}
		\hat v_x \\
		\hat v_y \\
		\hat v_z
	\end{bmatrix}
	=
	\left\Vert\begin{matrix}
		(6\vec f_s)\mod 1 -\nicefrac{1}{2} \\
		\vec f_t -\nicefrac{1}{2} \\
		\nicefrac{1}{2}
	\end{matrix}\right\Vert
	\cdot
	\begin{cases}
		\begin{bmatrix}\hat Z & -\hat X & -\hat Y\end{bmatrix}\transpose, & \text{if}\ i=0 \smallskip\\
		\begin{bmatrix}-\hat Z & \hat X & -\hat Y\end{bmatrix}\transpose, & \text{if}\ i=1 \smallskip\\
		\begin{bmatrix}\hat X & -\hat Y & -\hat Z\end{bmatrix}\transpose, & \text{if}\ i=2 \smallskip\\
		\begin{bmatrix}\hat X & \hat Z & -\hat Y\end{bmatrix}\transpose, & \text{if}\ i=4 \\
		\begin{bmatrix}-\hat X & -\hat Z & -\hat Y\end{bmatrix}\transpose, & \text{if}\ i=5
	\end{cases} \qed
\end{gather}
\end{subequations}
Cube-mapping perspective-map formula, where $\hat v$ represents visual sphere vector and $\vec f$ is the screen coordinate. $[\hat Z\quad-\hat X\quad-\hat Y]$ represents rotation matrix of each cube side. Operation $x\mod 1$ is equivalent to $\textbit{fract}(x)$ function.

\subsubsection[Multiple-screen array]{Multiple-screen array map}
\label{sub:screen array}

\begin{figure}[H]
	\centering\includegraphics{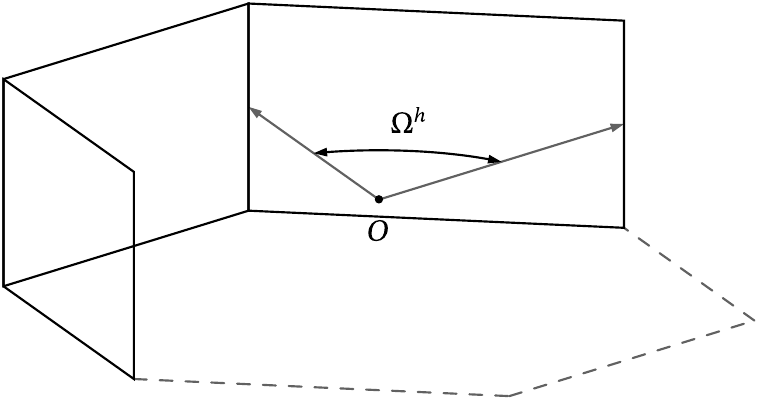}
	\caption[Screen array]
		{Screen array model with $\Omega^h=\nicefrac{2}{6}\,\pi$, and $O$ as observation point.}
	\label{fig:screen array}
\end{figure}

\begin{subequations}
\begin{align}
	i &= \left\lfloor n \vec f_s \right\rfloor +\frac{1-n}{2}
	\\
	\begin{bmatrix}
		\hat v_x \\
		\hat v_y \\
		\hat v_z
	\end{bmatrix}
	&=
	\left\Vert
	\begin{matrix}
		2\big(( n \vec f_s )\mod 1\big) -1 \\
		(2 \vec f_t -1)n \div a \\
		\cot\frac{\Omega^h}{2}
	\end{matrix}
	\right\Vert
	\begin{bmatrix}
		\cos(i\Omega^h) & 0 & \sin(i\Omega^h) \\
		0 & 1 & 0 \\
		-\sin(i\Omega^h) & 0 & \cos(i\Omega^h)
	\end{bmatrix} \qed
\end{align}
\end{subequations}
$n$-screen array perspective-map formula, where $n$ is the number of screens, $\hat v$ represents visual sphere vector, $\vec f$ is the screen coordinates vector, with $\Omega$ and $a$ being a single-screen AOV and aspect ratio, respectively.
% screen array function
\begin{lstlisting}[float=p,
	label={lst:screen array},
	caption={
		[Horizontal screen array map]
		Visual sphere vector $\hat v\in{[-1,1]^3}$ function from texture coordinates $\vec f\in[0,1]^2$ in GLSL, for $n$--screen-array perspective.
	}
]
vec3 getScreenArrayMap(vec2 tex_coord, int n, float horizontal_fov, float aspect)
{
	horizontal_fov = radians(horizontal_fov);
	float index = floor(n*tex_coord.s)-0.5*n+0.5;
	// Get single screen coordinates
	vec3 direction = normalize(vec3(
		2.0*fract(n*tex_coord.s)-1.0,
		n*(2.0*tex_coord.t-1.0)/aspect,
		1.0/tan(horizontal_fov*0.5)
	));
	// Rotate each screen individually
	horizontal_fov *= index;
	vec2 rotation = vec2(
		sin(horizontal_fov),
		cos(horizontal_fov)
	);

	return direction*mat3(
		 rotation[1], 0.0, rotation[0],
		         0.0, 1.0,         0.0,
		-rotation[0], 0.0, rotation[1]
	);
}
\end{lstlisting}

\subsubsection[Virtual reality]{Virtual Reality perspective map}
\label{sub:vr}

\begin{figure}[H]
	\centering\includegraphics{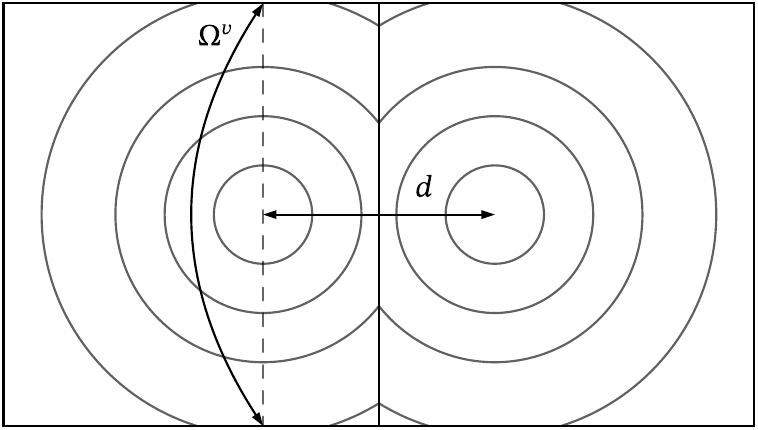}
	\caption[Virtual Reality]
		{Virtual Reality model with $d$ as the IPD distance relative to total screen width. Such that $d=\frac{\text{IPD}}{\text{width}}$.}
	\label{fig:vr}
\end{figure}

\begin{subequations}
\begin{align}
	i &= \text{sign}(\vec f_s-\nicefrac{1}{2})
	\\
	\begin{bmatrix}
		\vec f'_x \\
		\vec f'_y
	\end{bmatrix}
	&=
	\left(\begin{bmatrix}
		2\big((2 \vec f_s)\mod 1\big) -1 \\
		2 \vec f_t -1
	\end{bmatrix}
	+\begin{bmatrix}
		i(1-2d) \\
		0
	\end{bmatrix}\right)
	\begin{bmatrix}
		0.5a \\
		1
	\end{bmatrix}
	\smallskip\\
	\begin{bmatrix}
		\vec f''_x \\
		\vec f''_y
	\end{bmatrix}
	&=
	\begin{bmatrix}
		\vec f'_x \\
		\vec f'_y
	\end{bmatrix}
		\frac
		{\overbrace{
			1+k_1|\vec f'|^2 +k_2|\vec f'|^4 +\hdots +k_n|\vec f'|^{2n}
		}^\text{radial lens distortion}}
		{\underbrace{
			1 +k_1 +k_2 +\hdots +k_n
		}_\text{bounds normalization}}
	\\
	\begin{bmatrix}
		\hat v_x \\
		\hat v_y \\
		\hat v_z
	\end{bmatrix}
	&=
	\left\Vert \begin{matrix}
		\vec f''_x \\
		\vec f''_y \\
		\cot\frac{\Omega^v}{2}
	\end{matrix} \right\Vert \qed
\end{align}
\end{subequations}
Virtual reality perspective map formula, where $\vec f$ represents screen coordinates, $d$ is the interpupillary distance (IPD) in screen-width scale, $a$ is the screen aspect ratio, $k_1,k_2,\hdots,k_n$ represent lens-distortion coefficients. $\hat v$ is the visual sphere vector and $\Omega^v$ is the vertical AOV.
% VR function
\begin{lstlisting}[float=p,
	label={lst:vr},
	caption={
		[Virtual reality map]
		VR visual-sphere vector $\hat v\in{[-1,1]^3}$ from texture coordinates $\vec f\in[0,1]^2$ in GLSL. The \textbit{ipd} variable is expressed in screen-width scale.
	}
]
vec3 getVRMap(vec2 tex_coord, float ipd, float vertical_fov, float k1, float k2, float aspect)
{
	float index = sign(tex_coord.s-0.5);
	// Convert to stereoscopic coordinates
	tex_coord.xy = 2.0*vec2(fract(2.0*tex_coord.s), tex_coord.t)-1.0;
	// Convert to view coordinates with IPD offset (in screen-width ratio)
	tex_coord.x = (tex_coord.x+index*(1.0-2.0*ipd) )*0.5*aspect;
	// Apply lens distortion
	float R2 = dot(tex_coord, tex_coord); // Radius squared
	tex_coord *= (1.0+k1*R2+k2*R2*R2)/(1.0+k1+k2); // Vertical normalization

	// Output incident vector
	return normalize(vec3(
		tex_coord,
		1.0/tan(radians(vertical_fov)*0.5)
	));
}
\end{lstlisting}

\begin{rem}
	VR perspective map can also be combined side-by-side with a regular perspective map. In extended desktop environment, extra perspective map could provide a monitor preview of the VR content, without an additional render call.
\end{rem}

\subsection[Lens distortions]{Lens distortion of perspective picture}
\label{sub:lens distortions}

Creating perspective picture of a real optical system may require additional deformation of the vector data. Most commonly used algorithm for this purpose is the \emph{Brown-Conrady} lens distortion model.\supercite{Wang2008}

\begin{subequations}
\begin{align}
	r^2 &= \vec f \cdot \vec f
	\\
	\begin{split}
		\begin{bmatrix}
			\vec f'_x \\
			\vec f'_y
		\end{bmatrix} &=
			\begin{bmatrix}
				\vec f_x \\
				\vec f_y
			\end{bmatrix}
			+
			\overbrace{
			\left( k_1r^2 + k_2r^4 +\cdots+ k_nr^{2n} \right)
			\begin{bmatrix}
				\vec f_x \\
				\vec f_y
			\end{bmatrix}
			}^\text{radial distortion}
		\\
			&+
			\underbrace{
			\left(
				\begin{bmatrix}
					p_1 \\
					p_2
				\end{bmatrix}
				\cdot
				\begin{bmatrix}
					\vec f_x \\
					\vec f_y
				\end{bmatrix}
			\right)
			\begin{bmatrix}
				\vec f_x \\
				\vec f_y
			\end{bmatrix}
			}_\text{thin prism}
			+
			\underbrace{
			\begin{bmatrix}
				q_1r^2 \\
				q_2r^2
			\end{bmatrix}
			}_\text{decentering}
	\end{split} \qed
\end{align}
\end{subequations}
Where $r$ is the dot product of two $\vec f$ vectors.
$k_1,k_2$ and $k_n$ are the radial distortion coefficients.
$q_1$ and $q_2$ are the decentering coefficients.
Thin prism distortion coefficients are denoted by $p_1$ and $p_2$.
% lens distortion function
\begin{lstlisting}[float=p,
	label={lst:lens distortion},
	caption={
		[View coordinates lens distortion]
		Lens distortion function for transformation of texture coordinates $\vec f\in[0,1]^2$ to view coordinates $\vec f'\in\real^2$ in GLSL.
	}
]
vec3 applyLensDistortion(vec2 tex_coord, float k1, float k2, float p1, float p2, float q1, float q2, float aspect)
{
	float diagonal = inversesqrt(aspect*aspect+1.0);
	// Center coordinates
	tex_coord.xy = 2.0*tex_coord.st-1.0;
	// Correct aspect
	tex_coord.x *= aspect;
	// Normalize diagonally
	tex_coord *= diagonal;

	vec2 distort_coord = tex_coord;
	// Get radius squared
	float R2 = dot(tex_coord, tex_coord);
	// Radial distortion
	distort_coord *= k1*R2+(k2*(R2*R2)+1.0);
	// Thin prism distortion
	distort_coord += tex_coord*dot(vec2(p1, p2), tex_coord);
	// Decentering distortion
	distort_coord += vec2(q1, q2)*R2;

	// Lens-distorted view-coordinates
	return distort_coord;
}
\end{lstlisting}

\section[Variable no-parallax point]{No-parallax point mapping}
\label{sec:npp}

Real optical systems exhibit phenomenon known as the floating no-parallax point,\supercite{Littlefield2006NoParallaxPoint} where each incident vector originates from different position within the lens. Meaning that real optical systems are not complicit with pinhole projection model. In pinhole model all incident vectors share same crossing point.
Therefore to simulate optical projection, view position has to change accordingly to incident vector origin of the calibrated lens, or in contrary, visible point should move the opposite way.
In spherical lens, NPP offset is in $z$-direction and can be described as a product of a function $\text{parallax}(\theta)$. As it changes accordingly to an angle between incident vector and the optical axis. Offset value can be approximated by optical measurement of parallax miss-alignment (see figure \vref{fig:Floating-no-parallax-point}).

To calibrate lens distortion with floating NPP, first static NPP picture must be obtained. If camera lens does not produce such image, it can be derived from a sequence of images, each taken at different $z$ position (see subfigure \vref{fig:npp image sequence z travel}).
In such image composite, every point of the picture plane shares common incident vector origin (in world-space). From this, two lens maps can be derived, \emph{parallax offset map} from sequence's pixel $z$ position and \emph{perspective map} from composite image of calibration rig/chart.

To render picture with floating NPP, each 3D point must be transformed prior to rasterization. Transformation is done accordingly to the perspective map incident position and associated parallax offset value. Either by moving the point or the view position. Noticeably the parallax offset value can also be encoded in a graph (for simple spherical lenses), instead of a texture map.

Rendering with rasterization would produce approximate result as values in-between vertices are interpolated, not transformed. Therefore best quality floating-NPP result is to be expected from ray-tracing. In such case, offset of a ray-origin-position should be performed. An equivalent of visual-sphere offset (see subfigure \vref{fig:npp offset model}).
Going back to rasterization, point offset can be achieved through projection of the \emph{parallax offset} texture map onto the scene geometry. Each geometry point can be transformed in view-space $z$ direction by a value from projected \emph{parallax offset} texture.

\section{Appendix}
\label{sec:appendix}

Perspective picture visible inside the visual space gives some sense of immersion (e.g. picture, film, computer game) even without visual illusion.\supercite{Dixon1987paradox} That's because it is perceived as a visual symbol of an abstract point of view, through which not picture plane is seen, but depicted space's mind-reconstruction. The picture immersion does not break, as long as appearance of the objects do not exhibit too much deformation.
Perceiving abstract point of view invokes separation from the surrounding. To enhance immersion, environment stimuli is being reduced. In a movie theater, to uphold the immersion lights are turned off and silence is expected. Horror-gameplay session are usually played at night, to separate from safe-space of home. This approach focuses virtual presence on depicted space.
\begin{rem}
	On the opposite side, picture as an integral part of the surrounding can be categorized under the \emph{Trompe-l'{\oe}il} technique.\supercite{Wade1999FoolingEyes}
\end{rem}
	\noindent Through the perception of the picture, physical properties of depicted space and objects within it are estimated. Mind fills the gaps, as symbols are always simplified versions of the real thing.
\begin{example*}
	Cup of coffee laying on a table has only one side visible at once, but it can be assumed that the opposite one is there too, since shape of the cup is known. This is a blank information filled by the mind.
\end{example*}
	\noindent Physical objects have their physical properties, but their visual symbols have some physical properties too, like angular size, perspective, shadow, etc. Those visual properties give some information about physical. In case of perspective, visual properties give information about depicted space and about point of view.
	Since most of the time picture represents a point of view (e.g. film, video game, visualization), it is wise to consider subject's properties of vision when designing picture perspective. But instead of producing mechanical simulation, perspective should symbolize total sensory experience.\supercite{Argan1946BrunelleschiPerspective}
\begin{thm}
	To create immersive visual symbol of a visual space, it is necessary to use curvilinear perspective instead of a linear.
\end{thm}
\begin{proof}
	Geometry of human visual space contradicts linear perspective principle, as visual field extends beyond linear perspective angle of view. Linear perspective, based on a tangent of an angle, exhibits limit of 179.(9)8\textdegree\ of view. While visual field extends horizontally up to 220\textdegree\ for binocular vision.\supercite{Hueck1840vision}
\end{proof}

\begin{figure}[H]
	\centering
	\begin{subfigure}{162pt}
		\includegraphics{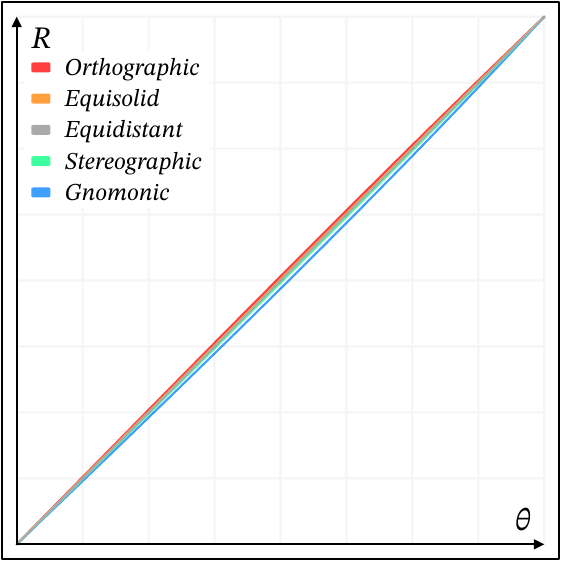}
		\caption{Graph plotting ray of angle $\theta\in[0\degree,\Omega]$ --- as the horizontal axis, and screen-position radius $R\in[0,1]$, as the vertical axis, where $\Omega=40\degree$ which is equivalent to $R=1$.}
		\label{fig:Graph-low-AOV}
	\end{subfigure}
	\hfill
	\begin{subfigure}{162pt}
		\includegraphics{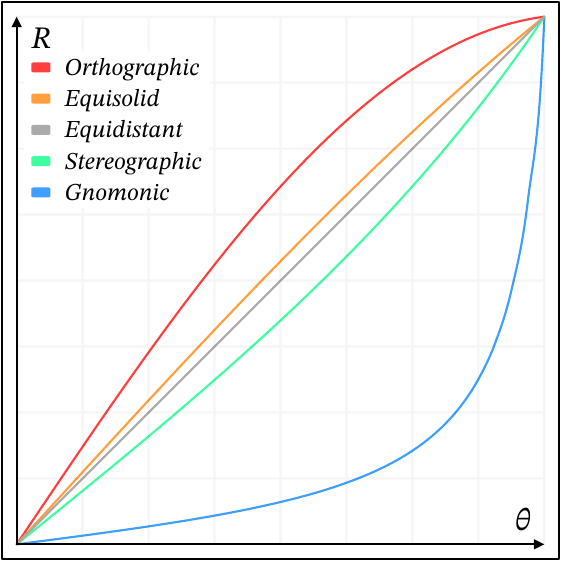}
		\caption{Graph plotting ray of angle $\theta\in[0\degree,\Omega]$ --- as the horizontal axis, and screen-position radius $R\in\left[0,1\right]$, as the vertical axis, where $\Omega=170\degree$ which is equivalent to $R=1$.}
		\label{fig:Graph-high-AOV}
	\end{subfigure}

	\caption[Radial compression charts]{Chart comparison of radial compression in five major azimuthal projections, across two different AOV ($\Omega$) values: narrow (\ref{fig:Graph-low-AOV}) and wide (\ref{fig:Graph-high-AOV}).}
	\label{fig:Radial-compression-chart}
\end{figure}

\noindent At narrow AOV both types of perspective are suitable for immersive picture. In such case differences between each projection are negligible. Subfigure \vref{fig:Graph-low-AOV} presents those differences in comparison to five major perspective projections. Same differences seem exaggerated at higher AOV values (see subfigure \ref{fig:Graph-high-AOV}).
\begin{cor}
	Practical limit for immersive picture in linear perspective is between 60\textdegree\ and 110\textdegree\ AOV. Wider-angles exhibit deformations known as the \emph{Leonardo Paradox}, \supercite{Dixon1987paradox} which are then dominant in the image perception and break picture's immersion.\qed
\end{cor}
\noindent To show wider-angle picture it is necessary to use curvilinear projection. But there is a tendency to see the world not through sight but understanding. We understand that the wall is flat, therefore we see it that way. Picture projected into the eye is just a visual symbol and has its own physical properties (e.g. perspective and shape). Therefore its visual representation is curvilinear, where the curvature symbolizes wider field of view. Reader can validate curvilinear nature of human visual space,\supercite{Baldwin_2014_VF_Perspective,Erkelens2015VSPerspective} by following \noun{A. Radley} experiment\supercite{Radley2014book}:

\begin{quote}
	\textitquote{Also when you have a moment, get a 30 cm ruler (...), and whilst looking forward bring it close to the bottom of your nose, and notice how its shape at the outer edges curves upwards and forwards. It may take you a few minutes to be able to see this effect, because you are so accustomed to not noticing it ! But once you do you will be amazed to see your curved field of view as it really is for the first time.} --- \noun{A. Radley}
\end{quote}

\subsection[Symbolic picture]{Visual space symbolic picture}
\label{sub:symbolic picture}

Since we came into conclusion that symbol of a visual space is curvilinear, there is a task of selecting between many non-linear projections. Each has properties that symbolize subject's perception or information span about depicted space. Not appearance of the symbol should dominate the picture, but projected information about point of view and depicted space properties.
\begin{problem}
	Which curvilinear perspective is best for visual symbol of visual space?
\end{problem}
\begin{prop}
	A model based on anamorphic lens geometry; a mix between fish-eye, panini and anamorphic projection.
	\begin{description}
		\item [fish-eye] as it can represent wider AOV than linear perspective (e.g. $\pi$) and conforms to the curvilinear nature of VS. Gives natural spatial awareness.
		\item [panini] to symbolize binocular vision; two spherical projections combined into one panoramic image.\footnote{Effect is also referred to as \emph{Stereopsis}.} Produces picture geometry more familiar to the viewer.
		\item [anamorphic] as cylindrical projection, like \emph{Panini}, tends to elongate proportions vertically; there is a need for correction. Correction should make object in focus proportional as it varies depending on position and size.
		\begin{rem}
			Only linear anamorphic correction will conform to the perspective picture definition (see \vpageref{def:conservation of perspective}).
		\end{rem}
	\end{description}
	\begin{rem}
		Equations in subsection \vref{sub:universal} (about perspective transformation) presents variables that drive all mentioned above geometrical factors.
	\end{rem}
\end{prop}

\subsection[Visual Sphere]{Visual sphere as a whole image}
\label{sub:visual sphere}

Common idea of an image is limited to a finite 2-dimensional plane. Which is subjective, due to constrains of human visual field and up-front placement of eyes. One can construct a rectangular frame, which at certain distance from the eyes will cover full visual field (VF). In case of some animals (e.g. horse, rabbit), visual space confines much wider VF. With only few blind spots, spanning to almost 360\textdegree AOV.\supercite{Murphy2010HorseVision,Bagley2013RabbitFOV} Such field cannot be enclosed by a single rectangular frame. Thus image nature is not of a frame. Another model has to be chosen instead. One able to cover full $\Omega=360\degree$ is a sphere.
\begin{rem}
	Cylindrical projection cannot cover full 360\textdegree AOV in all directions. It is a hybrid between frame and spherical model. When vertically-oriented, full $\Omega^v<180\degree$.
\end{rem}
All three-dimensional space around given observation point, can be projected onto a sphere, with given observation point as a origin. Even doe sphere itself is a 3D object, its surface (as well as image nature)\supercite{Rybczynski2009teatrise} is two-dimensional. Therefore creating perspective picture is a matter of representing portion of the visual sphere on a flat surface; a fundamental topic in cartography.
Concept of a sphere as a model of visual space goes back as far as 300 BC, where Greek mathematician \noun{Euclid} seem first to mention (others are \noun{L. Da Vinci} and \noun{F. Aguilonius}).\supercite{Tyler2009Euclid}
\begin{rem}
	Each projection of sphere onto a flat surface is a compromise and can preserve only some properties (e.g. shape, area, distance or direction), which in case of perspective picture relates to some symbolic information about physical space.
\end{rem}
\begin{defn}
	Let us define perspective picture as the azimuthal projection, where lines converging at optical axis vanishing point remain straight, that \emph{conservation of perspective} may occur (see definition \vpageref{def:conservation of perspective}).\label{def:azimuthal perspective}
\end{defn}

\subsubsection[Azimuthal projections]{Physical space properties preserved in azimuthal projections}
\label{sub:azimuthal projections}

Below are presented static properties of five major azimuthal projections. Properties of motion can be found in sub-subsection \vref{sub:perspective in motion}. It is important to know which symbolic information about space is carried in each perspective projection, so that design choice for perspective geometry may be conscious.
\begin{description}
	\item [Gnomonic] (rectilinear) projects all great circles as straight lines, thus preserving directions. For 3D projection, straight lines in object-space remain straight. It does not preserve proportions, angles nor area or distances (see subfigure \vref{fig:Rectilinear}). Extreme distortion occurs away from the center, in a form of radial stretch (see \emph{Leonardo Paradox})\supercite{Dixon1987paradox}. AOV $\Omega\in(0,\pi)$.
	\begin{example*}
		Most common perspective type in painting, 3D graphics and architectural visualization. Sometimes it is used to overemphasize building appearance by leveraging \emph{Leonardo Paradox}.\supercite{Dixon1987paradox} Wide AOV combined with lowered optical center creates an effect of acute corners. This produces an extraordinary look. Such visual-trick may confuse the public, as experience of symbolic picture won't match building visual-space appearance.
	\end{example*}
\end{description}
\begin{description}
	\item [Stereographic] (conformal) preserves angles (at line intersection point). There is no perceivable radial compression, thus smaller figures retain their shape. It does not preserve distances (non-isometric), nor angular surface area. For 3D projection, most important factor is the conservation of proportions (see subfigure \vref{fig:Stereographic}). AOV $\Omega\in(0,2\pi)$.
	\begin{example*}
		In a picture with stereographic projection, face of the actor keeps its shape and proportions, even at wide AOV. This projection also gives best spatial-awareness sensation (where visual cues are available).
		Good use case is navigation through tight spaces and obstacles.
	\end{example*}
\end{description}
\begin{description}
	\item [Equidistant] preserves angular distance from the center point (see subfigure \vref{fig:Equidistant}). For 3D projection, angular speed of motion is preserved. Radial compression remains low-to-moderate at extreme $\Omega$ angles. AOV $\Omega\in(0,2\pi]$.
	\begin{example*}
		This projection is recommenced for target aiming or radar map navigation, where all targets are projected onto a Gaussian Sphere.
	\end{example*}
\end{description}
\begin{description}
	\item [Equisolid] preserves angular area. Gives good sensation of distance (see subfigure \vref{fig:Equisolid}). Radial compression is moderate up to $\pi$. Near maximum $\Omega$, compression is high. AOV $\Omega\in(0,2\pi]$.
	\begin{example*}
		When there are no spatial cues, this is best projection for putting emphasis on the distance to the viewer.\supercite{Glaeser1999CurvedPerspectiveVR} Good use case is flight simulation, where only sky and other aircraft are in-view.
	\end{example*}
\end{description}
\begin{description}
	\item [Orthographic] preserves brightness. It is a parallel projection of a visual hemisphere. Has extreme radial compression, especially near $\pi$ (see subfigure \vref{fig:Orthographic}). AOV $\Omega\in(0,\pi]$.
	\begin{example*}
		Most commonly found in very cheap lenses, like the peephole door viewer.
		Thanks to brightness being proportional to occupied image area, it found common use in sky photography and scientific research.\supercite{Nikon2020,Thoby2012FisheyeLens}
	\end{example*}
\end{description}

\subsection[Motion and perspective]{Image geometry and sensation of motion}
\label{sub:motion}

Picture perspective affects the way motion is perceived. It can enhance certain features, like proportions and shapes, movement or spatial-awareness. It can also guide viewer attention to a specific region of image (e.g. center or periphery). Knowledge about those properties is essential for conscious image design.

% perspective in motion
\begin{figure}[H]
	\centerline{
		\begin{subfigure}[t]{212pt}
			\includegraphics{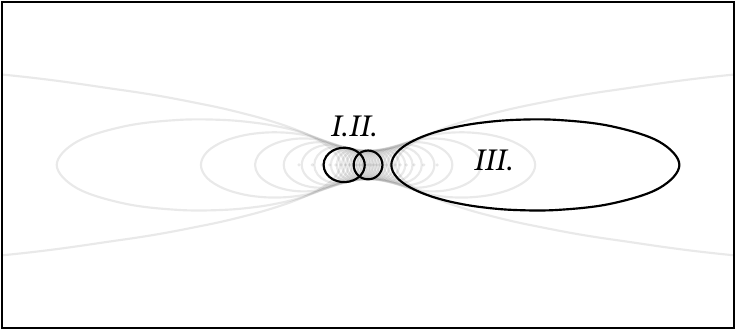}
			\caption{Rectilinear (Gnomonic) projection}
			\label{fig:Rectilinear}
		\end{subfigure}
		\hspace{0.382em}
		\begin{subfigure}[t]{212pt}
			\includegraphics{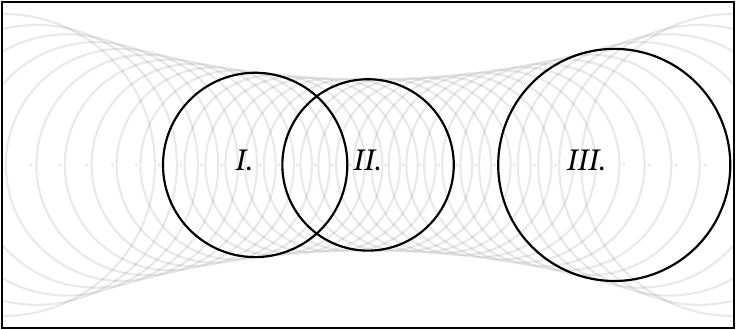}
			\caption{Stereographic projection}
			\label{fig:Stereographic}
		\end{subfigure}
	}\vspace{1em}
	\centerline{
		\begin{subfigure}[t]{212pt}
			\includegraphics{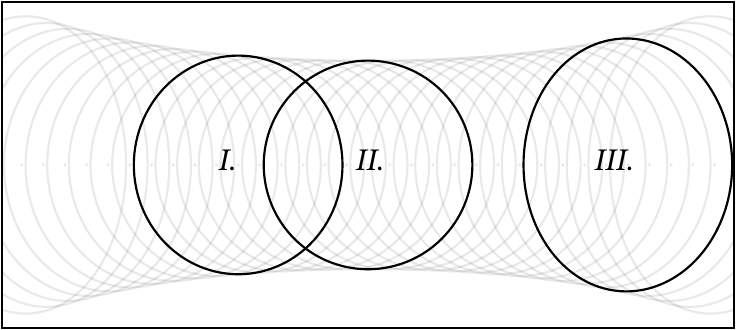}
			\caption{Equidistant projection}
			\label{fig:Equidistant}
		\end{subfigure}
		\hspace{0.382em}
		\begin{subfigure}[t]{212pt}
			\includegraphics{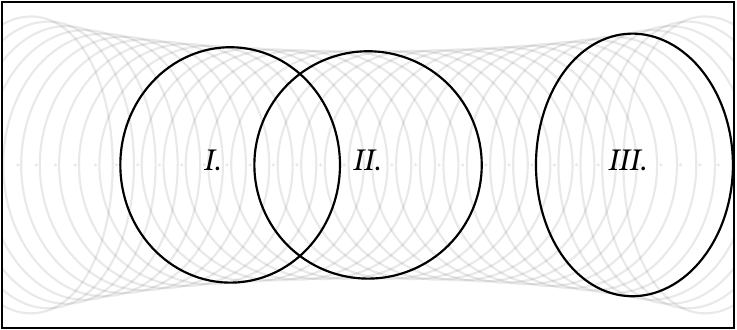}
			\caption{Equisolid projection}
			\label{fig:Equisolid}
		\end{subfigure}
	}\vspace{1em}
	\centerline{
		\begin{subfigure}[t]{212pt}
			\includegraphics{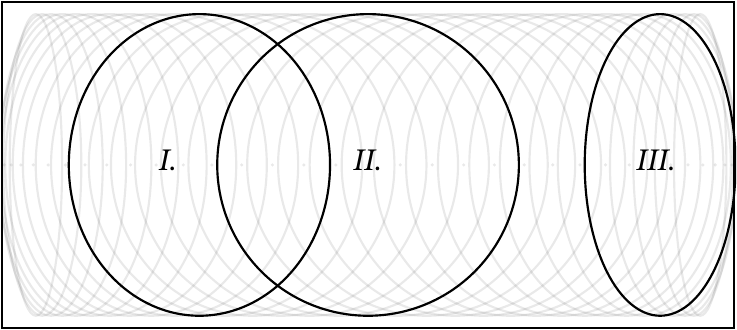}
			\caption{Orthographic-azimuthal projection}
			\label{fig:Orthographic}
		\end{subfigure}
		\hspace{0.382em}
		\begin{subfigure}[t]{212pt}
			\includegraphics{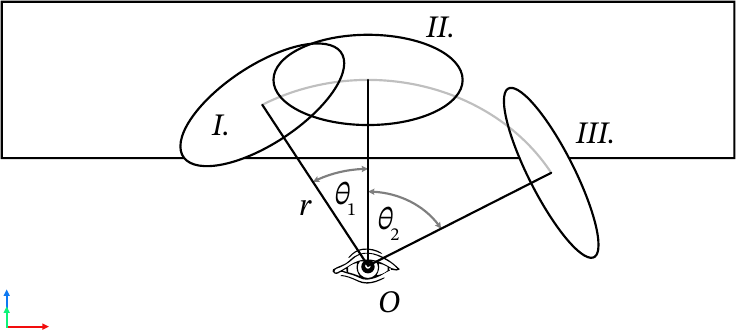}
			\caption{Model of examples; an equal-size discs ($I$, $II$, $III$) at equal distance $r$, from and facing observation point $O$, where $\Omega^h=170\degree$, $\theta_1=30\degree$ and $\theta_2=60\degree$}
			\label{fig:Model}
		\end{subfigure}
	}
	\caption[Motion in various azimuthal projections]{Example of motion in perspective picture of various azimuthal projections. Here subfigure \ref{fig:Model} presents viewed elements layout.}
	\label{fig:Perspective-example}
\end{figure}

\subsubsection{Attention focusing}
\label{sub:attention focusing}

In film design, there are several techniques that focus viewer attention on a specific portion of the picture (like \emph{motion}, \emph{color}, \emph{light} and \emph{composition}). Attention focusing through composition and motion is related to picture's perspective, as its geometry can compress and stretch the image. In composition, \emph{rule of thirds} states that viewer attention focuses on four corners of a rectangle produced by division of an image into three, equally-sized rows and columns. In motion, attention generally drives toward objects approaching the camera or those growing in scale. Attention also focuses on objects entering image frame. Same rules apply loosely in reverse, as attention suspense.
Filmmakers tend to frame the image so that region of interest lays in accordance to the \emph{rule of thirds}.
In case of computer games, region of interest is usually located at the very center, thus viewer must overcome the \emph{principle of thirds} and some properties of \emph{linear perspective} in order to switch attention to that region. In order to focus on the center, games usually incorporate some non-diegetic elements, like crosshair. Such approach may lower immersiveness of symbolic picture.\supercite{Casamassina2005KingKong}

\subsubsection[in motion]{Attention focusing motion in perspective}
\label{sub:perspective in motion}

Radial stretching and compression are the main attention focusing factors of perspective projection. They give subconscious sensation of movement towards camera, and can amplify figure screen-relative speed of motion.
\begin{description}
	\item [{Gnomonic}] (rectilinear), due to extreme radial stretch attention drives towards periphery. When approaching image bounds figure grows in scale and speed (see subfigure \vref{fig:Rectilinear}). This combined with motion-sensitive peripheral vision adds to the effect. At wider AOV amplified motion breaks immersion of symbolic picture.

	\item [{Stereographic}] also draws attention towards periphery. Figures grow in scale near bounds, but immersion does not break, as proportions are preserved, even at wide AOV (see subfigure \vref{fig:Stereographic}).

	\item [{Equidistant}] drives attention toward center, as figures in periphery are radially compressed (see subfigure \vref{fig:Equidistant}). This projection preserves screen-relative, radial speed of motion, making it uniform and representative across the picture.

	\item [{Equisolid}] also drives attention toward the center, as radial compression is even greater (see subfigure \vref{fig:Equisolid}). Figure speed of motion in screen-space slightly declines towards periphery.

	\item [{Orthographic}] has extreme radial compression that breaks immersion of a symbolic picture (see subfigure \vref{fig:Orthographic}). When in motion, image seem to be imposed on an artificial sphere.
\end{description}
\emph{Gnomonic} and \emph{Orthographic} projections are the two extremes of azimuthal spectrum. They both are least suited for an immersive picture. \emph{Cylindrical} perspective type, while symbolizing binocular vision, also gives visual cue for the vertical axis orientation. Such cue is undesirable in case of camera roll-motion, or when view is pointing up/down, as image vertical axis will then not be aligned with depicted space orientation. In such case perspective geometry should transition to \emph{spherical} projection (see figure \vref{fig:Example-motion}).

% perspective type for camera orientation
\begin{figure}[H]
	\centerline{
		\begin{subfigure}[t]{162pt}
			\includegraphics{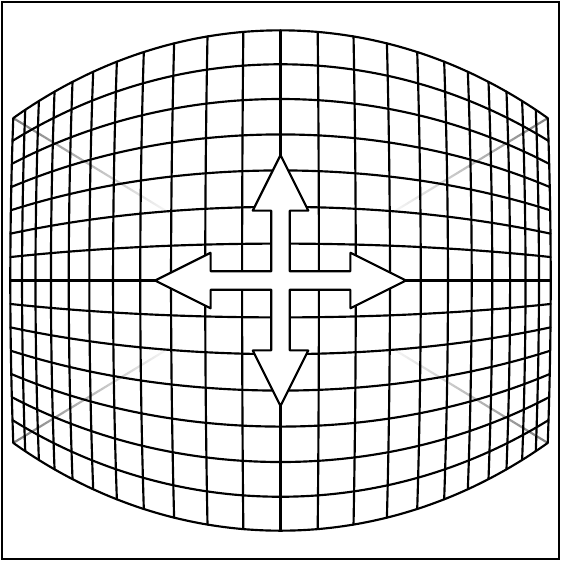}
			\caption[Pitch and yaw motion image geometry example]
				{Example of image geometry for pitch/yaw motion (arrows), where $\Omega^h=120\degree$, $k=0$, $l=10\%$ and $s=95\%$.}
			\label{fig:Pitch-yaw-motion}
		\end{subfigure}
		\hfill$\leftrightarrow$\hfill
		\begin{subfigure}[t]{162pt}
			\includegraphics{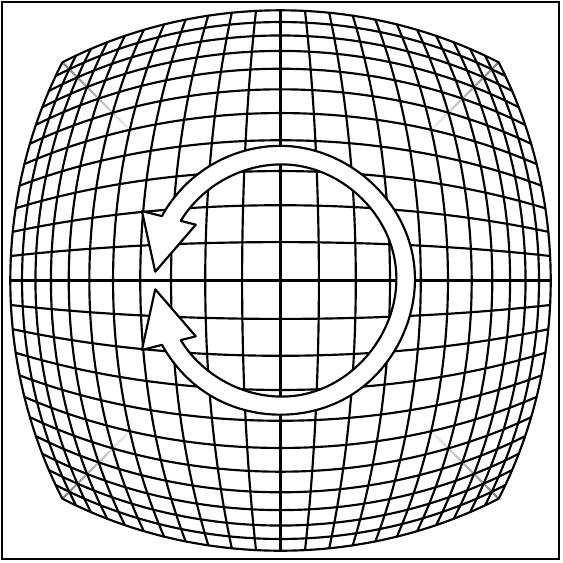}
			\caption[Roll motion image geometry example]
				{Example of image geometry for roll motion (arrows), where $\Omega^h=120\degree$, $k=0$, $l=100\%$ and $s=95\%$.}
			\label{fig:Roll-motion}
		\end{subfigure}
	}
	\caption[View orientation and perspective]
		{Examples of image geometry for a given view-motion type.}
	\label{fig:Example-motion}
\end{figure}

\subsection[History of the topic]{History of the topic and previous work}
\label{sub:history}

Current image abstract theorem was established in 15\textsuperscript{th} century book \emph{De Pictura,} by \noun{L. B. Alberti}. Based on invention of \noun{F. Brunelleschi}, \noun{Alberti} defined geometrical and theoretical rules for perspective projection design.\supercite{Argan1946BrunelleschiPerspective} These rules are currently used in polygon-based CG graphics.
Major theoretical statement that laid foundation on image projection technology and present understanding of image nature can be traced back to \noun{Alberti} abstract definition of image. He described painting being like a window in a wall\supercite{Alberti1972teatrise}:
\begin{quote}
	\textitquote{First of all, on the surface on which I am going to paint, I draw a rectangle of whatever size I want, which I regard as an open window through which the subject to be painted is seen.} --- \noun{L. B. Alberti}
\end{quote}
But in times of its discovery, as for now, linear perspective introduced itself with several issues. When there's a need for a wide-angle view, one close to human visual field, geometrical distortions appear to dominate visual aspect of the picture. These issues were noticed by Renaissance artists, like \noun{L. Da Vinci}. He put to test the Alberti Theorem and produced paintings of accurate-perspective.\supercite{Vinci1480painting} In his \emph{Treatise on Painting}, \noun{Da Vinci} notes that picture conforms to the idea of a window only when viewed from one specific point in space.\supercite{Vinci2014teatrise} Stating that seen otherwise, objects appear distorted, especially in the periphery. Picture then, viewed from a different point ceases to be like a window in a wall and becomes a visual symbol of an abstract point of view.\footnote{Effect also referred to \emph{Zeeman Paradox}.\supercite{Dixon1987paradox}}
Some 18\textsuperscript{th} century late Baroque and Neoclassical artists, when encountered mentioned issues, introduced derivative projections. Like \noun{G. P. Pannini} with later re-discovered \emph{Panini Projection},\supercite{Pannini1754painting} or \noun{R. Barker}, who established the term\emph{ Panorama}.\supercite{Barker2019wiki} This was a new type of perspective. In a form of cylindrical projection, where abstract window frame becomes horizontally curved, reducing deformation artifacts in wide, panoramic depictions of architecture.

% Anamorphic lens example
\begin{figure}[H]
	\centering\includegraphics{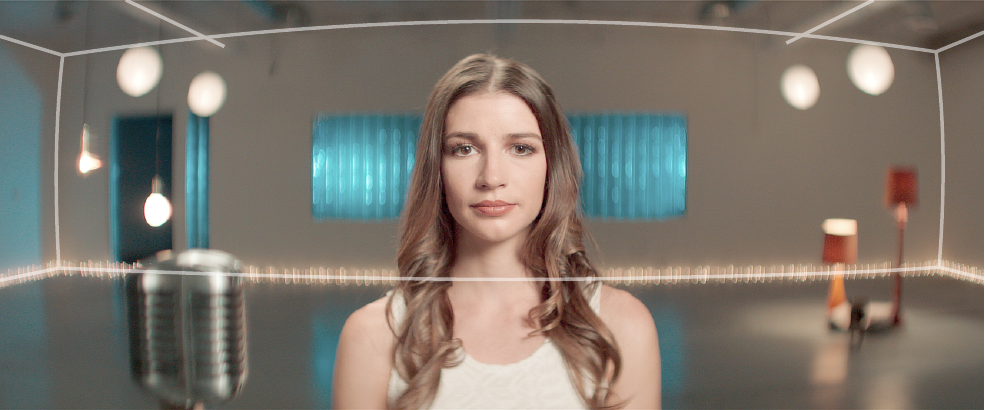}
	\caption[Anamorphic lens picture example]
		{Still from \noun{Todd--AO High-Speed Anamorphic} lens (35mm T1.4) with visible curvilinear perspective. This type of lens was featured in films like \emph{Conan the Barbarian}, \emph{Dune} and \emph{Mad Max}.\supercite{ShareGrid2017ToddAO}\hfill\copyright\ 2017 ShareGrid, Inc.}
	\label{Fig.Anamorphic}
\end{figure}

\paragraph{Invention of motion picture}
followed by the rise of film industry, resulted in a demand for new image geometry. Previously still, now pictures had to be pleasing to the eye, in motion. 1950s brought anamorphic cinematography to the wider audience. Lenses like \noun{CinemaScope} and later \noun{Panavision\supercite{Konigsberg1987book}} became a standard in film production.
Figure \vref{Fig.Anamorphic} shows example of mixed spherical and cylindrical projection with perspective preservation, of anamorphic lens.
\begin{defn}
	\label{def:conservation of perspective}
	\emph{Conservation of perspective} - lines converging at the optical-axis vanishing-point remain straight.
	\begin{rem}
		See also perspective picture definition \vpageref{def:azimuthal perspective}.
	\end{rem}
\end{defn}
\noindent
CG image technology did not follow film industry in that field. Still based on \noun{Alberti} theorem, computer graphics became incompatible with film, generating great costs, when two had to be joined together.\supercite{Rydzak2014interview} Stitching such picture requires lens aberration rotoscopy, where geometry correction is performed manually at each frame.
Currently in computer-games industry, CG imagery is practically unable to produce realistic, curvilinear simulation of visual space (VS), or even simulate anamorphic lens geometry, due to limits of linear perspective and resource costs of overcoming those issues. Some hybrid solutions where proposed,\supercite{Gascuel2008proceedings,Glaeser1999CurvedPerspectiveVR} that combine rasterization with ray-tracing, or tessellation. Such approach allows for a semi-practical and limited production of real-time pictures in a non-linear perspective.

\section[Conclusion]{Conclusion and future work}
\label{sec:conclusion}

Visual sphere perspective model expands possibilities for image creation. Like a vision of a classical artist is richer and more dynamic than his final creation, so should be a model describing image. So much that virtual vision would have to be reduced to fit the medium, with a room for adjustment.

In medieval times people were fascinated with mirror, it depicted reality as it really is, a task impossible for human hand. But mirror could not reflect the vision of imagination, so much as painting. It did not produce realistic reflection of reality, nor reflection of imagination. In art there was a pursuit involving philosophy and religion to reach some imaginary reflection. Only after \noun{Brunelleschi} experiment proved that painting created by human hand, with strict rules of geometry, reflected from a mirror fits well into reality, artists started pursuing mathematics of art. Film and photography became a new mirror, but still not one of imagination. Finally computer graphics became a format capable to reflect one's vision of imagination, but it struggled to be photo-real.

Language of image is ultimately a language of symbols, it works in conjunction with mind, which recreates visual space of the picture.
Computer game image is special in a sense, that rather being a viewing glass (like in a film) it impersonates point of view of the protagonist, which invokes different symbolic measures.
Real-time graphics should be capable to incorporate such visual symbols. One way for achieve this is to change current perspective model for a bigger one.
Proposed here change includes the use of universal perspective model which defines simple variables for manipulation of perspective geometry. It also includes production technique for a real-time imagery using this model.
Such new solution fits well into current graphics pipeline, replacing only low-level rasterization processes. Aliasing-free end result is comparable in quality to 8$\times$MSAA and enables previously impossible visual effects in real-time graphics.
Presented visual sphere model unites all types of perspective projections under one technical solution, making perspective a fluid construct. Perspective vector-maps can be easily combined and transformed by spherical interpolation.
Picture geometry can now be designed to smoothly adapt to the visual story, giving new dimension of control over mental perception of image. Presented concepts and equations may also find their use in other fields, not rendering-related.
This is a great base as well as complete solution upon which many new technologies and works can emerge. Some specific use cases still require additional research, like hidden surface removal and no-parallax point mapping.
Further studies will include research over calibration and simulation of real optical systems with floating NPP. Also performance tests, comparison to current solutions should be evaluated by future research.
Psychological analysis of perspective geometry magnitude of influence on depicted space perception, performed on a large sample data, could be an interesting field to study.

\section*{\centering Acknowledgments} % unnumbered unlisted section
\addcontentsline{toc}{section}{Acknowledgments} % add to table of contents and bookmarks
\label{sec:acknowledgments}

Special thanks to
\noun{Zbigniew Rybczyński} who told that we don't see the world in rectilinear perspective and that CG image does not match camera lens projection (during his 2010 conference in Cieszyn). This single presentation became a spark for this research.
Additional thanks to \noun{Shaun Williams} who combined cylindrical and spherical projection based on a view angle.
To \noun{Alan Radley} and \noun{Kim Veltman} for their interest in this topic and paper.
To \noun{Arash Shiva} who provided permission for use of the anamorphic lens test.
Final thanks to \noun{Dobiesław Kaczmarek} who taught principles of perspective drawing.
And to all unmentioned individuals who gave much needed feedback, it helped greatly improving this paper.

%%%%%%%%%%%%%%%%%
%% End of file %%
%%%%%%%%%%%%%%%%%

\input{./appendix.tex}

%% file: appendix.tex
%%%%%%%%%%%%%%%%%%%%%%%%
%% Extra figures page %%
%%%%%%%%%%%%%%%%%%%%%%%%

\input{./figures_page.tex}

%%%%%%%%%%%%%%%%%%
%% Bibliography %%
%%%%%%%%%%%%%%%%%%

\pagebreak
\addcontentsline{toc}{section}{References}
\label{sec:references}
\printbibliography

%%%%%%%%%%%%%
%% License %%
%%%%%%%%%%%%%

\begin{figure}[b]
	% \section*{License}
	\phantomsection
	% Add entry to the table of contents and bookmarks
	\addcontentsline{toc}{section}{License notice}
	% add target to the license icon on title page
	\hypertarget{license}{
		% Legal notice
		\noindent\footnotesize\copyright\ {\the\year} Jakub Maksymilian Fober (the Author).\quad The Author provides this document (the Work) under the Creative Commons CC BY-ND 3.0 license available online. To view a copy of this license, visit \href{\license}{\license} or send a letter to Creative Commons, PO Box 1866, Mountain View, CA 94042, USA. The Author further grants permission for reuse of \hyperref[pg:title]{images and text from the first page} of the Work, provided that the reuse is for the purpose of promoting and/or summarizing the Work in scholarly venues and that any reuse is accompanied by a scientific citation to the Work.\vspace{4pt} \\
		% Creative Commons BY-ND 3.0 logo
		\centerline{
			\href{\license}{\includegraphics{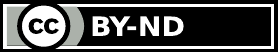}}
		}
	}
\end{figure}

\end{document}

%% file: figures_page.tex
%%%%%%%%%%%%%%%%%%
%% Figures page %%
%%%%%%%%%%%%%%%%%%

% rasterization examples
\begin{figure}[p]
	\vspace{-2em}\centering
	\begin{subfigure}{5.8cm}
		\includegraphics[width=5.8cm]{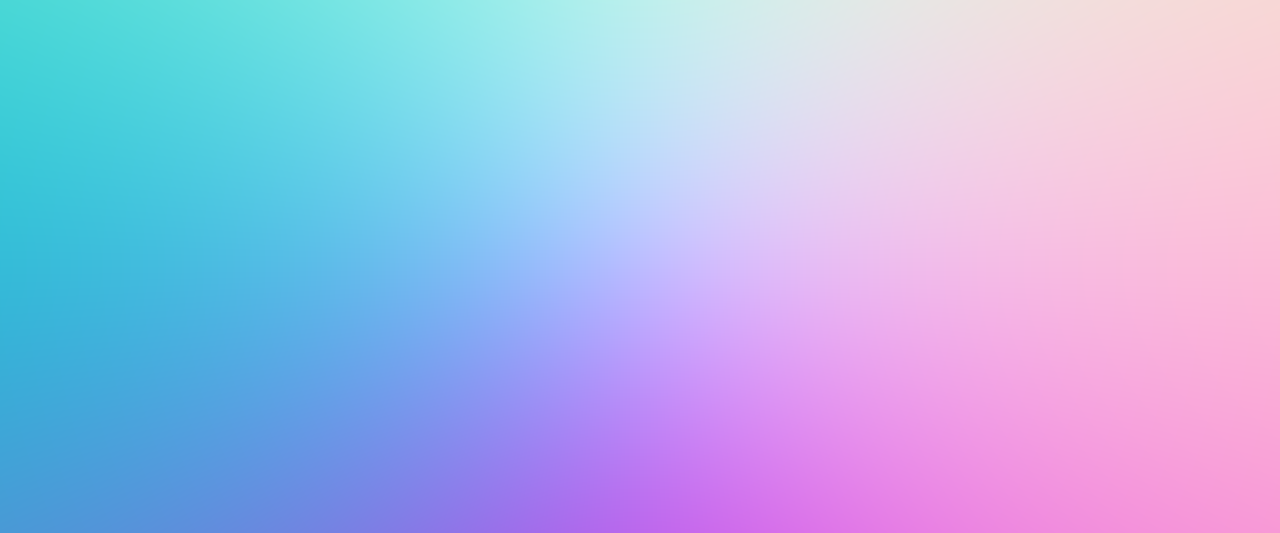}
		\caption{\emph{Vector map} of rectilinear perspective, where $\Omega^d=140\degree$, $k=1$, $l=1$.}
		\label{sub:raster-rectilinear map}
	\end{subfigure}
\;$\rightarrow$\;
	\begin{subfigure}{5.8cm}
		\includegraphics[width=5.8cm]{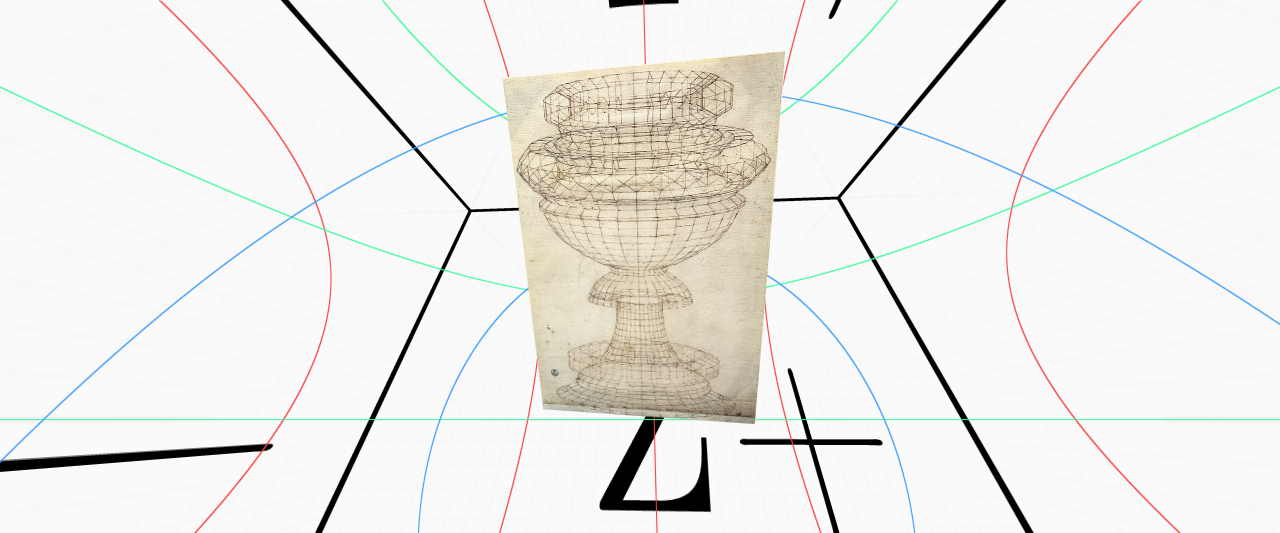}
		\caption{\emph{Rasterized quad} in rectilinear perspective, where $\Omega^d=140\degree$, $k=1$, $l=1$.}
		\label{sub:raster-rectilinear polygon}
	\end{subfigure}
\\[0.618em]
	\begin{subfigure}{5.8cm}
		\includegraphics[width=5.8cm]{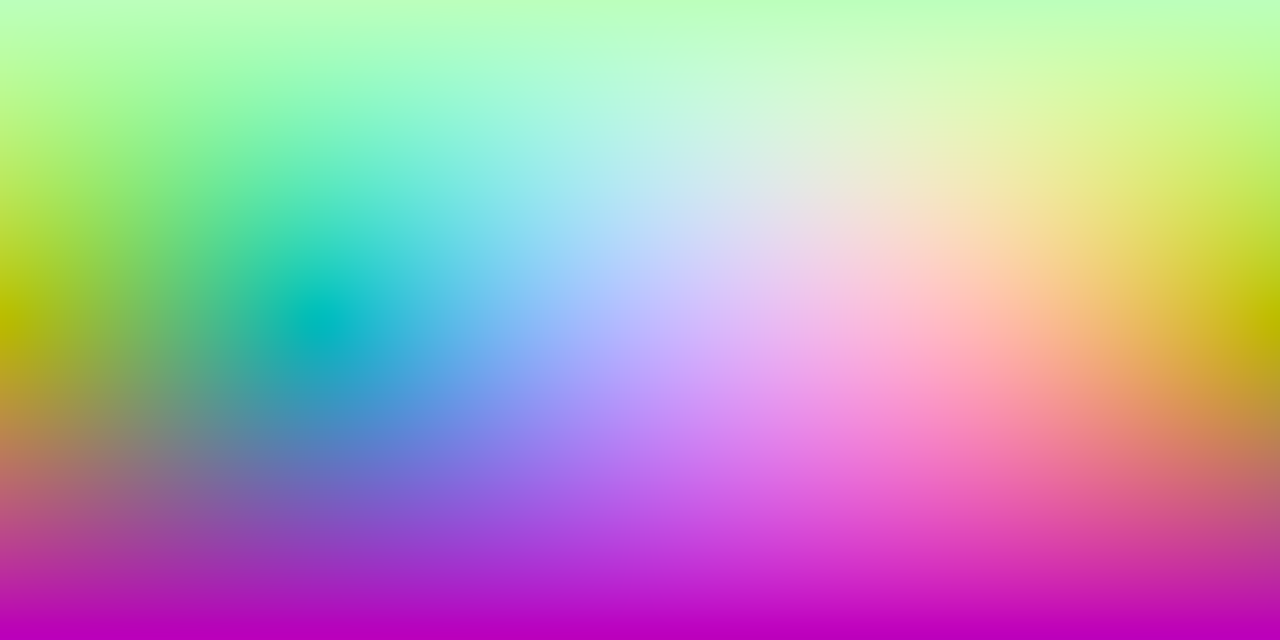}
		\caption{Equirectangular projection \emph{vector map} of a whole sphere, where $\Omega^h=360\degree$ and $\Omega^v=180\degree$.}
		\label{sub:raster-equirectangular map}
	\end{subfigure}
\;$\rightarrow$\;
	\begin{subfigure}{5.8cm}
		\includegraphics[width=5.8cm]{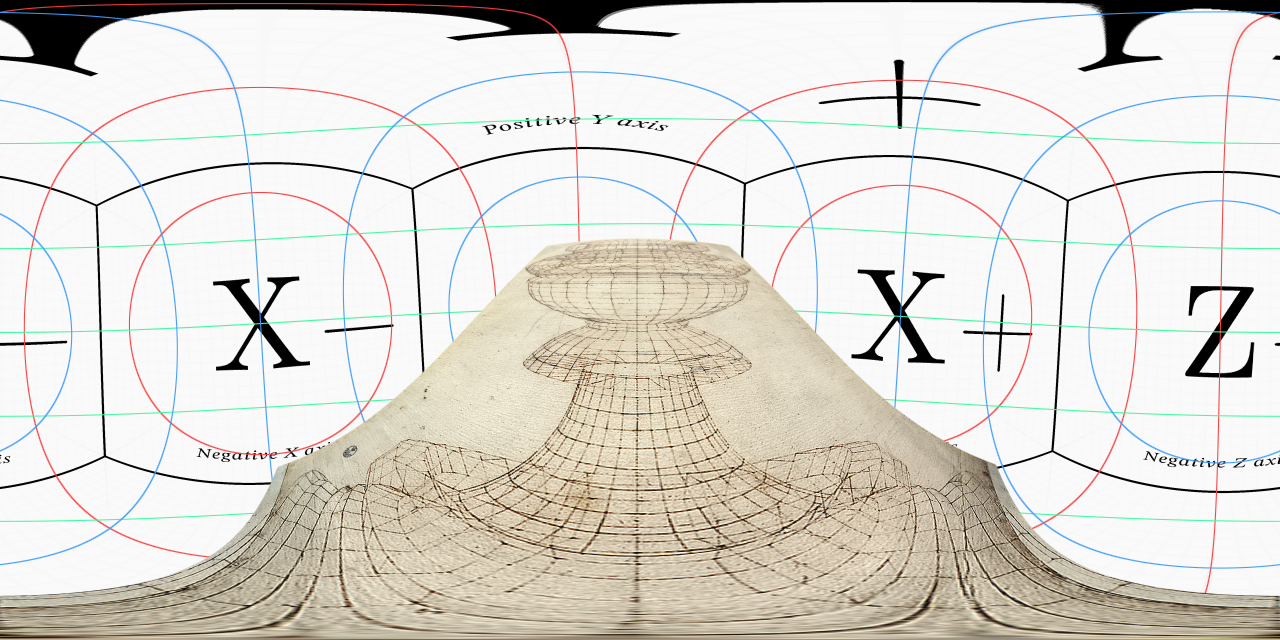}
		\caption{\emph{Rasterized quad} in equirectangular projection, where $\Omega^h=360\degree$ and $\Omega^v=180\degree$.}
		\label{sub:raster-equirectangular polygon}
	\end{subfigure}
\\[0.618em]
	\begin{subfigure}{5.8cm}
		\includegraphics[width=5.8cm]{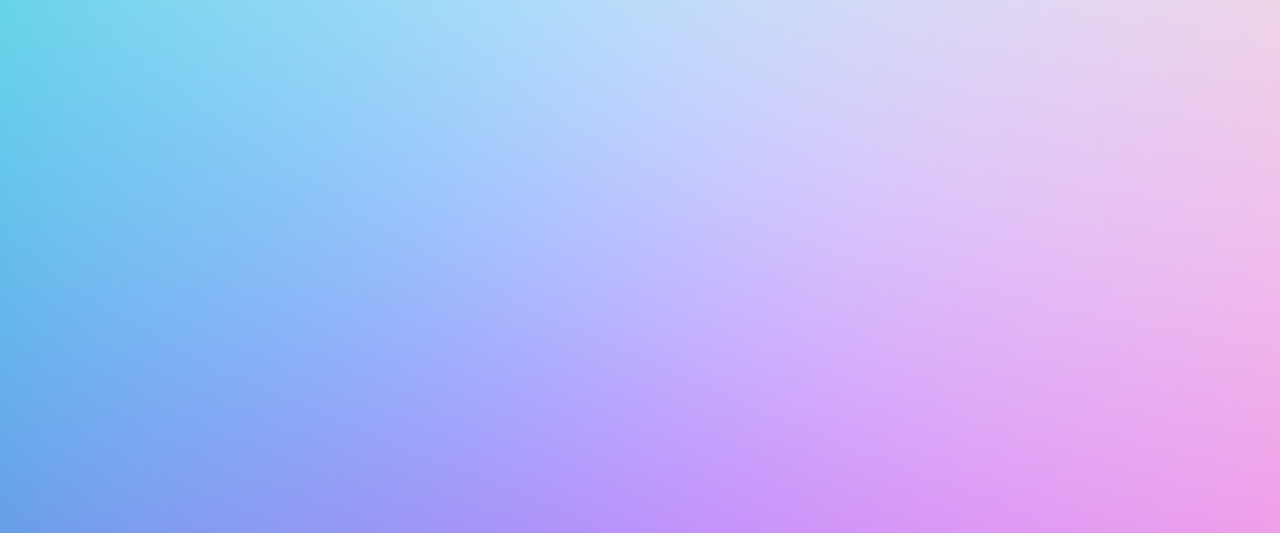}
		\caption{\emph{Vector map} of mustache-type lens distortion, where $\Omega^d=131\degree$, $k=0.32$, $l=62\%$, $s=86\%$, $k_1=-0.6$, $k_2=0.4$.}
		\label{sub:raster-131 map}
	\end{subfigure}
\;$\rightarrow$\;
	\begin{subfigure}{5.8cm}
		\includegraphics[width=5.8cm]{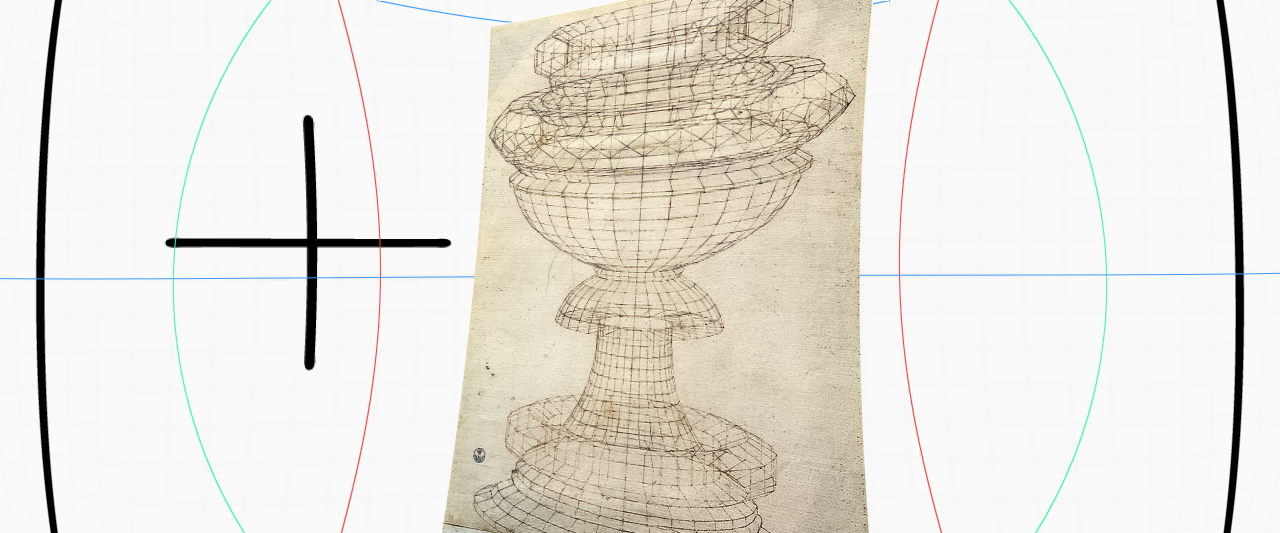}
		\caption{\emph{Rasterized quad} in mustache-type lens distortion, where $\Omega^d=131\degree$, $k=0.32$, $l=62\%$, $s=86\%$, $k_1=-0.6$, $k_2=0.4$.}
		\label{sub:raster- polygon}
	\end{subfigure}
\\[0.618em]
	\begin{subfigure}{5.8cm}
		\includegraphics[width=5.8cm]{Fig_fov270_k0_32_l0_62_s0_86_map}
		\caption{\emph{Vector map} of curvilinear fish-eye perspective, where $\Omega^d=270\degree$, $k=0.32$, $l=62\%$, $s=86\%$.}
		\label{sub:raster-270 map}
	\end{subfigure}
\;$\rightarrow$\;
	\begin{subfigure}{5.8cm}
		\includegraphics[width=5.8cm]{Fig_fov270_k0_32_l0_62_s0_86_triangle}
		\caption{\emph{Rasterized quad} in curvilinear fish-eye perspective, where $\Omega^d=270\degree$, $k=0.32$, $l=62\%$, $s=86\%$.}
		\label{sub:raster-270 polygon}
	\end{subfigure}
\\[0.618em]
	\begin{subfigure}{5.8cm}
		\includegraphics[width=5.8cm]{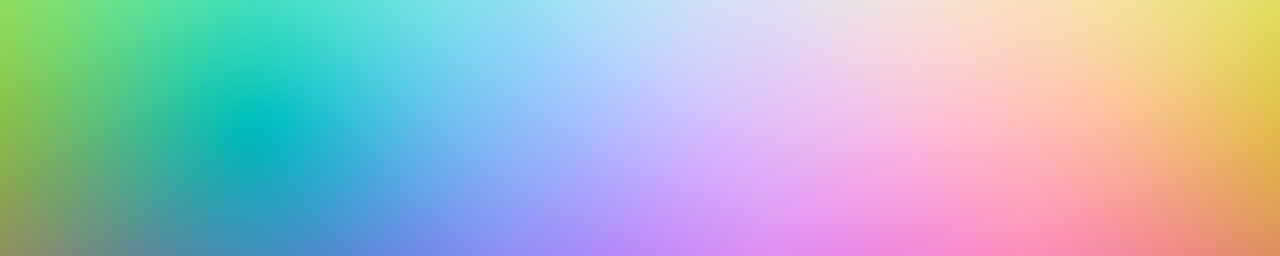}
		\caption{\emph{Vector map} of five-screen horizontal array in rectilinear projection, where single screen $\Omega^{h,v}=60\degree$, $k=1$ and $l=1$, giving a total $5\Omega^h=300\degree$.}
		\label{sub:raster-array map}
	\end{subfigure}
\;\raisebox{2em}{$\rightarrow$}\;
	\begin{subfigure}{5.8cm}
		\includegraphics[width=5.8cm]{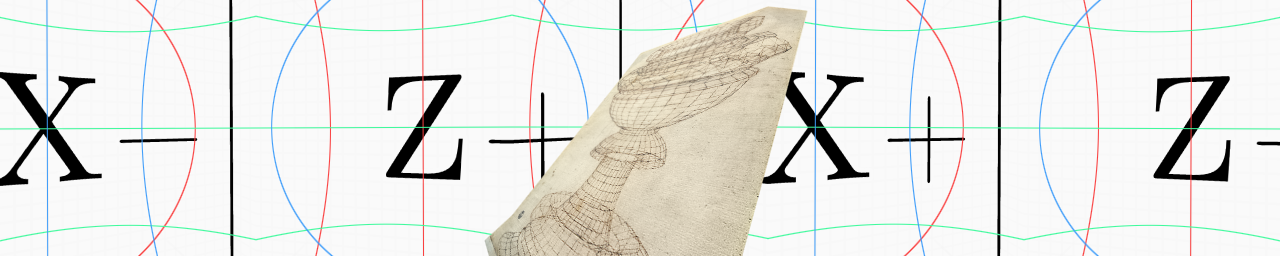}
		\caption{\emph{Rasterized quad} in five-screen horizontal array, where single screen $\Omega^{h,v}=60\degree$, $k=1$ and $l=1$, giving a total $5\Omega^h=300\degree$.}
		\label{sub:raster-array polygon}
	\end{subfigure}

	\caption[Rasterization examples]
		{Examples of polygon quad rasterization directly from three-dimensional space to image, using visual-sphere vector map $G$, where $G_{\vec{st}}\in[0,1]^3\mapsfrom[-1,1]^3$.}
	\label{fig:Examples-of-rasterization}
\end{figure}

% rasterization flowchart
\begin{figure}[p]
	\vspace{-3em}
	\centering\includegraphics{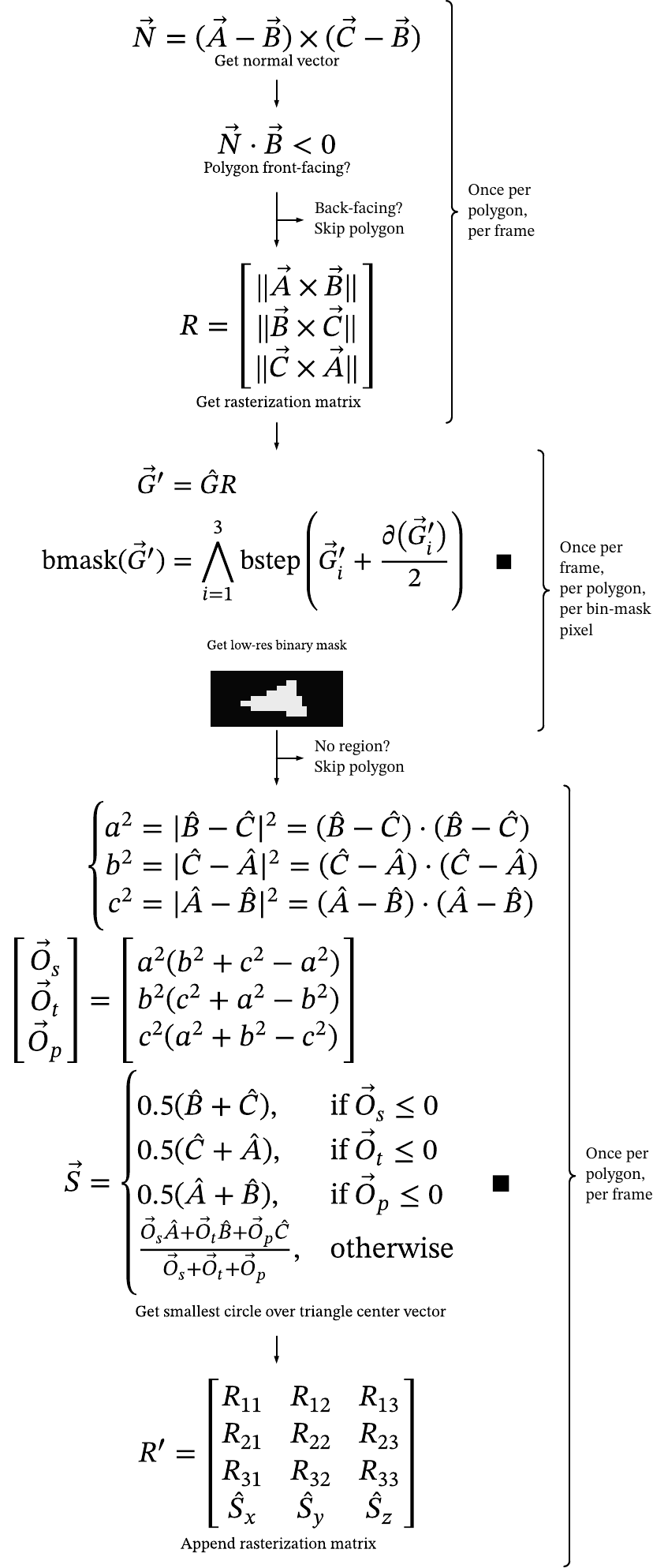}
	\caption[Vertex rasterization process flowchart]
		{Vertex rasterization process flowchart. Chart starts with $\vec A,\vec B,\vec C$ triangle points (transformed to camera-space) and produces low-resolution binary mask (using perspective vector map $G$) and $4\times3$ rasterization matrix $R'$.}
		\label{fig:vertex flowchart}
\end{figure}
\begin{figure}[p]
	\centering\includegraphics{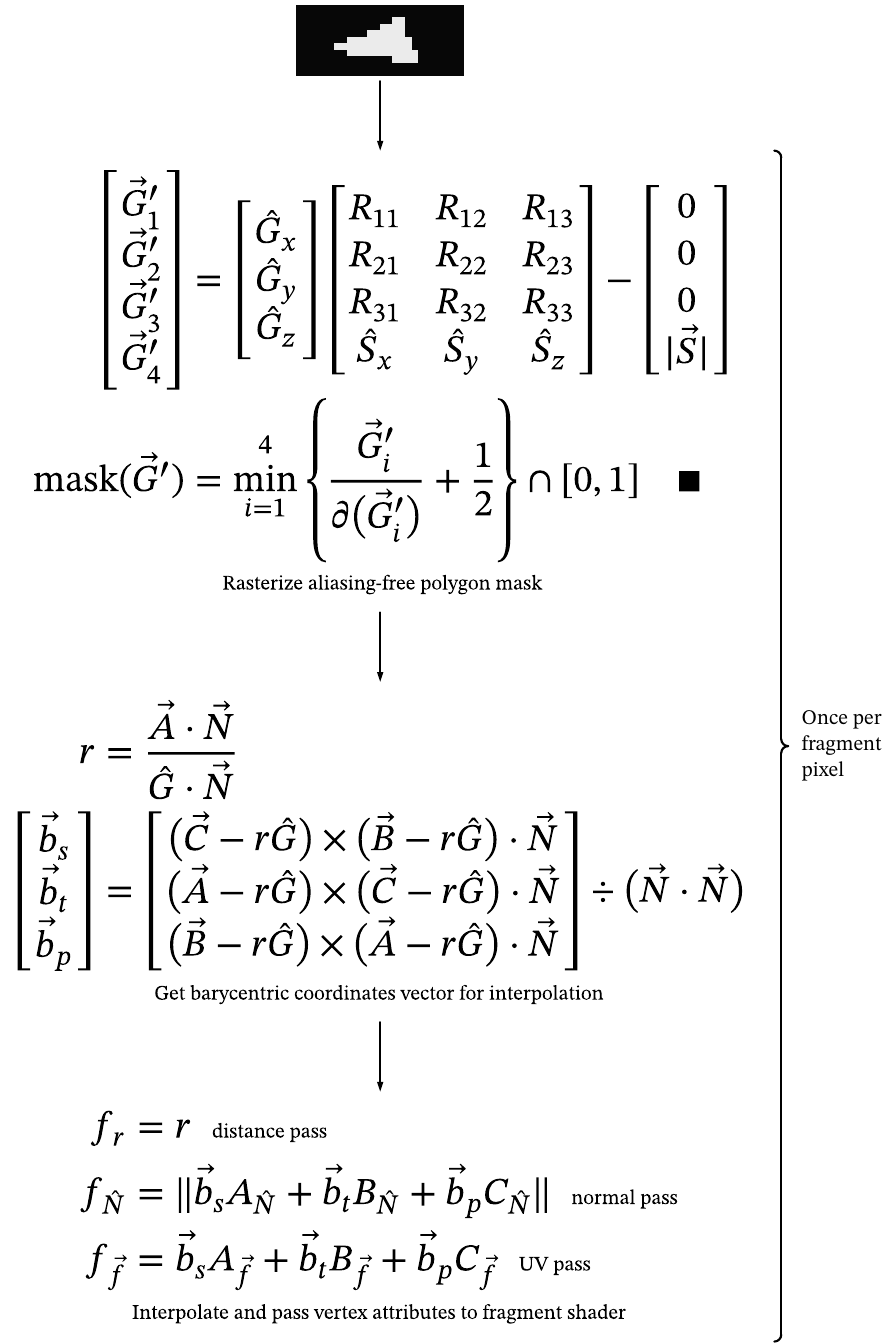}
	\caption[Fragment rasterization process flowchart]
		{Fragment rasterization process flowchart. Chart starts with rasterization matrix $R'$, binary mask of render regions, polygon normal vector $\vec N$, triangle points $\vec A,\vec B,\vec C$ and perspective vector map $G$. For each pixel, aliasing-free mask is produced and barycentric coordinates $\vec b$ are calculated, then interpolated vertex data $f$ is passed to fragment shader as output.}
		\label{fig:fragment flowchart}
\end{figure}

% floating no-parallax point
\begin{figure}[p]
	\vspace{-2em}
	\centerline{
		\begin{subfigure}{212pt}
			\includegraphics{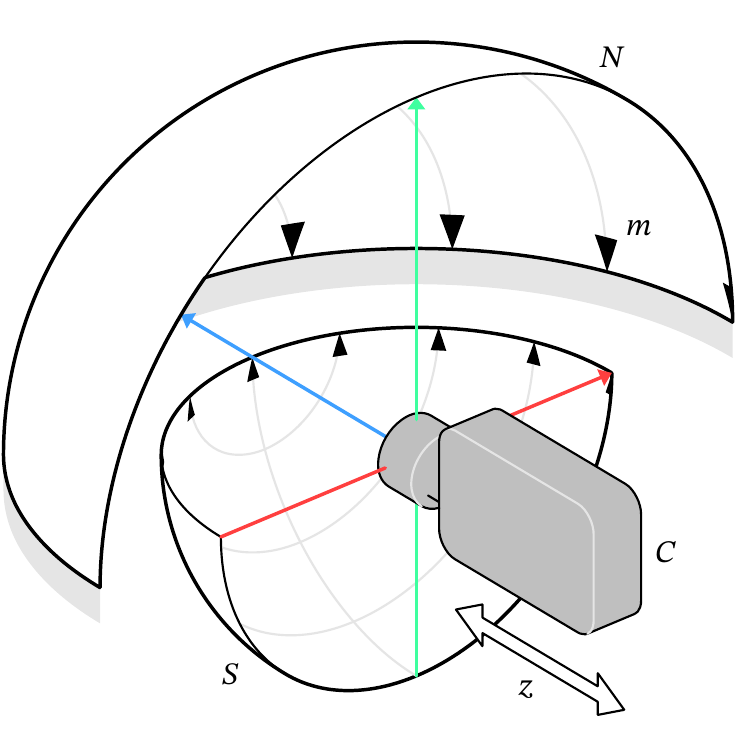}
			\caption[NPP model]
				{Model of no-parallax point (NPP) calibration rig. Which measures misalignment of $m$ markers between northern $(N)$ and southern hemisphere ($S$), as seen through camera $C$.}
			\label{fig:npp model}
		\end{subfigure}
		\quad
		\begin{subfigure}{212pt}
			\includegraphics{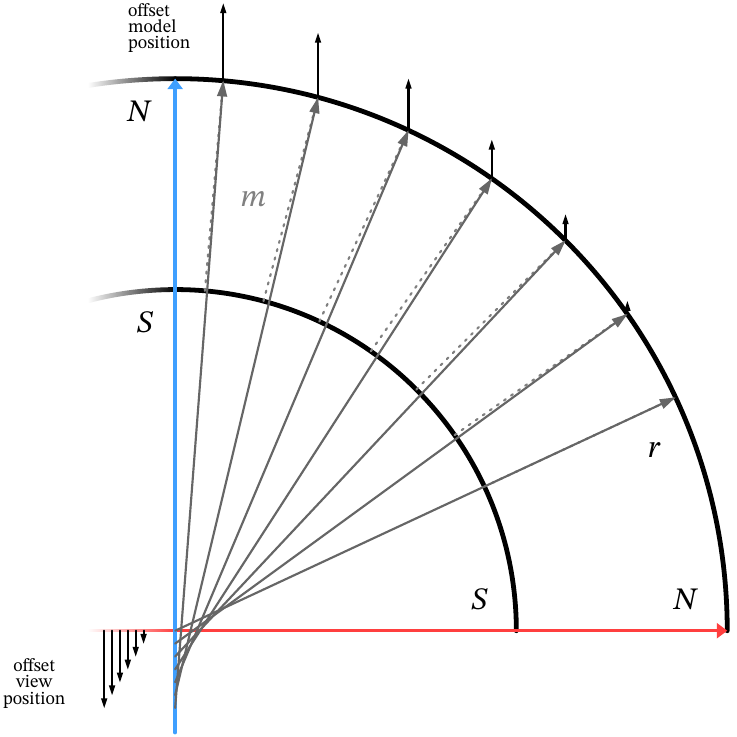}
			\caption[NPP examples]
				{Simulation of floating NPP in fish-eye camera. Black arrows represent position offset of the model for rasterization and view position offset for ray-tracing. When measuring the NPP (sub-figure \ref{fig:npp misalignment examples}), camera travel is opposite to view position offset.}
			\label{fig:npp offset model}
		\end{subfigure}
	}
	\centerline{
		\begin{subfigure}{212pt}
			\includegraphics{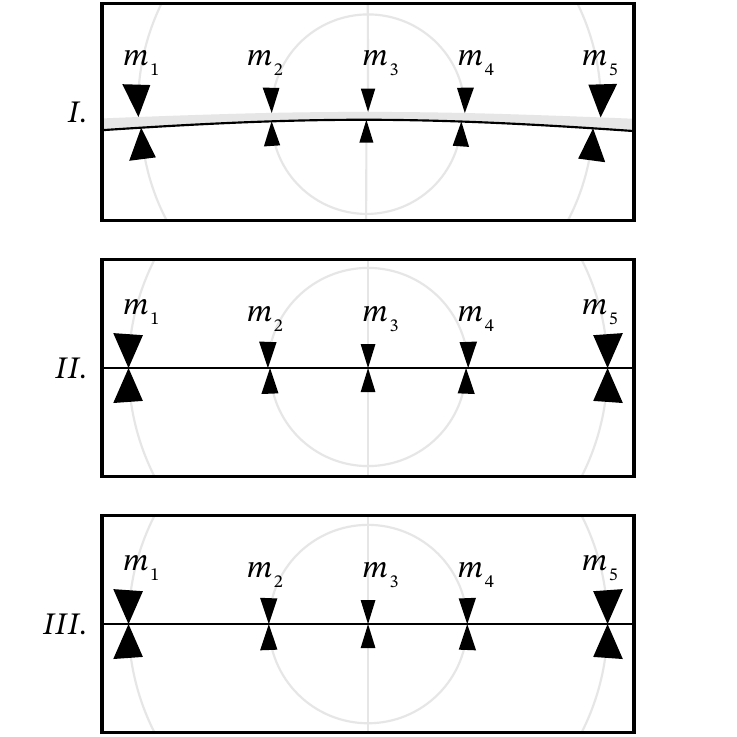}
			\caption[NPP examples]
				{Example $I$ presents misalignment of the camera in all three axes. Example $II$ presents alignment in X, Y axis, where peripheral marker pair $m_1$, $m_5$ are aligned near horizontal AOV, while pairs $m_2$, $m_4$ are misaligned due to floating NPP. Example $III$ presents "slit-scan" composite of variable $z$ position, where all markers are aligned.}
			\label{fig:npp misalignment examples}
		\end{subfigure}
		\quad
		\begin{subfigure}{212pt}
			\includegraphics{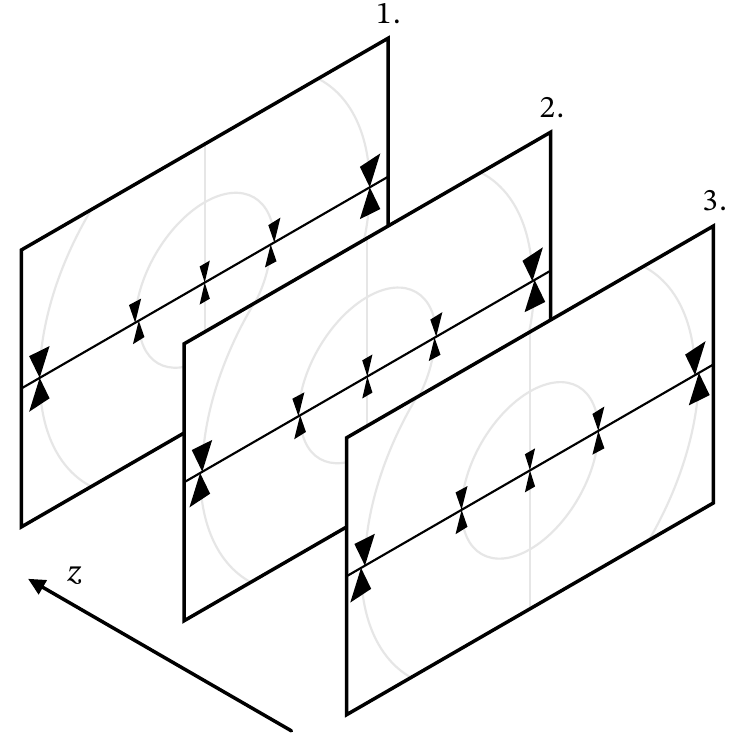}
			\caption[NPP image sequence of $z$ travel]
				{Example of a variable camera position $z$ encoded in an image sequence $(1,2,3)$. Element $1$ presents alignment of the peripheral markers, while element $3$, alignment of the side markers. Element number $2$ presents position in-between.}
			\label{fig:npp image sequence z travel}
		\end{subfigure}
	}
	\caption[Floating no-parallax point]
		{Floating no-parallax point rendering and measurement process.}
	\label{fig:Floating-no-parallax-point}
\end{figure}